\documentclass{statsoc}
\pdfoutput=1 
\usepackage[a4paper]{geometry}
\usepackage{graphicx}
\usepackage[textwidth=8em,textsize=small]{todonotes}
\usepackage{amsmath, amsfonts, amssymb}
\usepackage{natbib}
\usepackage[T1]{fontenc}
 \usepackage{stmaryrd} 
\usepackage[bbgreekl]{mathbbol}

\newcommand{\cp}{\ensuremath{X}}
\newcommand{\R}{\ensuremath{\mathbb{R}}}

\renewcommand{\tt}[1]{{#1}'}
\newcommand{\inv}[1]{{#1}^{-1}}

\newcommand{\babs}[1]{\llbracket #1 \rrbracket}
\newcommand{\bbabs}[1]{\bigg\llbracket #1 \bigg\rrbracket}
\newcommand{\abs}[1]{\left\| #1 \right\|}
\newcommand{\tabs}[1]{\| #1 \|}

\newcommand{\e}[1]{\text{e}^{#1}}

\newtheorem{thm}{Theorem}
\newtheorem{lemma}{Lemma}

\newtheorem{remark}{Remark}

\newtheorem{standAss}{Standing Assumption}

\title[Diffusions with mixed effects]{Multivariate inhomogeneous diffusion models with covariates and mixed effects}
\author[Gro\ss e Ruse  {\it et al.}]{Mareile Gro\ss e Ruse}
\address{University of Copenhagen, Copenhagen, Denmark}
\email{mareile@math.ku.dk}

\author{Adeline Samson}
\address{Laboratoire Jean Kutzmann, Universit\'{e} Grenoble-Alpes, Grenoble, France}

\author[M. Gro\ss e Ruse, A. Samson and S. Ditlevsen]{Susanne Ditlevsen}
\address{University of Copenhagen, Copenhagen, Denmark.}

\begin{document}

\begin{abstract}
Modeling of longitudinal data often requires diffusion models that incorporate overall time-dependent, nonlinear dynamics of multiple components and provide sufficient flexibility for subject-specific modeling. This complexity challenges parameter inference and approximations are inevitable. We propose a method for approximate maximum-likelihood parameter estimation in multivariate time-inhomogeneous diffusions, where subject-specific flexibility is accounted for by incorporation of multidimensional mixed effects and covariates. We consider $N$ multidimensional independent diffusions $X^i = (X^i_t)_{0\leq t\leq T^i}, 1\leq i\leq N$, with common overall model structure and unknown fixed-effects parameter $\mu$. Their dynamics differ by the subject-specific random effect $\phi^i$ in the drift  and possibly by (known) covariate information, different initial conditions and observation times and duration. The distribution of $\phi^i$ is parametrized by an unknown $\vartheta$ and $\theta = (\mu, \vartheta)$ is the target of statistical inference.  Its maximum likelihood estimator is derived from the continuous-time likelihood. We prove consistency and asymptotic normality of $\hat{\theta}_N$ when the number $N$ of subjects goes to infinity using standard techniques and consider the more general concept of local asymptotic normality for less regular models. The bias induced by time-discretization of sufficient statistics is investigated. We discuss verification of conditions and investigate parameter estimation and hypothesis testing in simulations.

\end{abstract}


\section{Introduction}

Many physical and biological processes recorded over time exhibit time-varying and non-linear dynamics.  There are two key demands for  reliable statistical investigation of such data: First, the model should enable a sufficiently comprehensive description of the dynamics, which translates into a minimum requirement on the model complexity. Often, this is well captured by ordinary differential equations (ODEs) with a suitable degree of non-linearity and dimensionality. Second, the model should be sufficiently parsimonious to facilitate robust statistical estimation. This parsimony can to some degree be achieved by including mixed effects \citep{pinheiro2006mixed, davidian2003nonlinear}: While assuming an overall model structure for all experimental units, some parameters are allowed to vary across the population to capture also individual-specific characteristics. Hence, it is not surprising that ODE models with mixed effects have become a  popular tool for inference on longitudinal data, \citep{lindstrom1990nonlinear, wolfinger1993laplace, tornoe2004non, guedj2007maximum, wang2007algorithms, ribba2014review, lavielle2014mixed}. This framework has, however, one important deficiency: The deterministic nature of ODE models does not capture uncertainties in the model structure and this can lead to biased estimates and false inference. For example, if the data have a periodic component with fluctuations in the phase, least squares estimation from a deterministic ODE will lead to predicted dynamics that are close to constant, and other estimation methods have to be sought for \citep{ditlevsen2005parameter}. 
This shortcoming can be addressed by replacing ODEs with stochastic differential equations (SDEs), thereby facilitating a more robust estimation \citep{donnet2010bayesian, moller2010, leander2014stochastic}. Furthermore, SDE models with mixed effects use data more efficiently. 
Consider an SDE model based on one longitudinal measurement. It is well-known that for fixed time horizon the drift estimator is inconsistent \citep{kessler2012statistical}. In many applications one has observations of several experimental units at hand, but dynamics seem too subject-specific to assume that all individuals share the same parameter values. This prohibits the otherwise natural approach to reduce the bias by pooling the data. However, when dynamics are individual-specific but structurally similar, the unknown drift parameter can be modeled as a mixed effect, which may reduce bias considerably. In combination with mixed effects, nonlinear, time-inhomogeneous and multidimensional SDEs thus become a highly versatile tool for the intuitive and comprehensive modeling of complex longitudinal data, allowing for more robust statistical inference. This framework of stochastic differential mixed-effects models (SDMEM) lends itself to numerous applications and thereby opens up for new insights in various scientific areas. \looseness=-1

Combining the benefits that are specific to mixed-effects and SDE models, however, entails particular challenges in terms of statistical inference. The key challenge lies in the intractability of the data likelihood, which now has two sources: The likelihood for (nonlinear) SDE models (given fixed parameter values) is  analytically not available, rendering parameter inference for standard SDE models a nontrivial problem in itself. This intractable quantity has then to be integrated over the distribution of the random effects and one realizes that numerical or analytical approximations are inevitable. 
%
%
The likelihood in SDE models can be approximated in various ways. Given discrete-time observations, the likelihood  is expressed in terms of the transition density. Approximation methods for the latter reach from solving the Fokker-Planck equation numerically \citep{lo1988maximum}, over  standard first-order (Euler-Maruyama) or higher-order approximation schemes and simulation-based approaches \citep{pedersen1995new, durham2002numerical} to a closed-form approximation via Hermite polynomial expansion \citep{ait2002maximum}. If continuous-time observations are assumed (e.g., if high-frequency data is available), transition densities are not needed and the likelihood can be obtained from the Girsanov formula \citep{phillips2009maximum}.
Popular analytical approximation techniques for nonlinear mixed-effects models are first-order conditional estimation (FOCE)  \citep{beal1981estimating} and Laplace approximation \citep{wolfinger1993laplace}.  A computational alternative to analytical approximation is the  expectation-maximization (EM) algorithm, or stochastic versions thereof  \citep{delyon1999convergence}.
\looseness=-1

In the  context of SDMEMs, the above mentioned approximation methods have been combined in various ways, depending on whether observations are modeled in discrete or continuous time and with or without measurement noise. 
Models for discrete-time data including measurement noise require marginalization of the likelihood over both the state and the random effects distribution. 
Approximation via FOCE in combination with the extended Kalman filter has been pursued by several authors \citep{tornoe2005stochastic, overgaard2005non, mortensen2007matlab, klim2009population, leander2014stochastic, leander2015mixed}. The most general setting was considered by \cite{leander2015mixed}, who allowed drift and diffusion function of the multivariate state SDE to depend on time, state and the individual parameters. However, theoretical convergence properties are not available in this setting. 
%
\cite{donnet2008parametric} approximate  a one-dimensional SDE by the Euler-Maruyama scheme, and employ a stochastic approximation EM algorithm to avoid marginalization entirely.  To circumvent the computationally expensive simulation of the SDE solution, \cite{delattre2013coupling} consider only the random effects as latent and use a Metropolis-Hastings algorithm to simulate these conditioned on the observation. Finally, a Bayesian setting for a one-dimensional homogeneous diffusion is considered by \cite{donnet2010bayesian}. 
In models for discrete-time observations without observation noise,  \cite{ditlevsen2005mixed}, consider one-dimensional linear SDEs with linear mixed effects, where the likelihood is available in closed form. \cite{picchini2010stochastic} and \citet{picchini2011practical} approximate the transition density by Hermite expansion and explore  Gaussian quadrature algorithms and Laplace's approximation to compute the integral over the mixed effect. Mixed effects that enter the diffusion coefficient are investigated by \cite{delattre2015estimation}.
Continuous-time observations are the starting point of investigations in \cite{delattre2013maximum}. They consider a univariate SDMEM without measurement noise for Gaussian mixed effects, which enter the drift linearly. \looseness=-1

None of the previously mentioned works provide theoretical investigations of the estimators, when the state process is modeled by a multivariate, time-inhomogeneous and nonlinear SDE. However, many biological and physical processes are time-varying and require a certain degree of model complexity, such as non-linearity and multidimensional states. Furthermore, none of the cited works include covariates. Especially for practitioners, being used to regression analyses, not including covariate information in a model seems highly restrictive.\looseness=-1 

The purpose of this article is two-fold. On the one hand, we extend the setting in \cite{delattre2013maximum} to a multidimensional state process with non-linear, time-inhomogeneous dynamics. We obtain an integral expression for the  likelihood. If the drift function for the SDE is linear in the random effect and if the random effects are   independent and identically $\mathcal{N}(0,\Omega)$-distributed with unknown covariance matrix $\Omega$, the integral expression for the  likelihood can be solved explicitly. If also the fixed effect enters the drift linearly, the likelihood turns into a neat expression, in which all remaining model complexity (multidimensionality of the state, nonlinearity, covariates) is conveniently hidden in the sufficient statistics. Under standard, but rather strict regularity conditions, we derive in subsection \ref{asymptotics_iid} the consistency and asymptotic normality of the MLE of the fixed effect and $\Omega$. We  also investigate the discretization error that arises when replacing the continuous-time statistics by their discrete-time versions. All results here can be shown using the same techniques as in \cite{delattre2013maximum}, however, they are tedious to write down due to the more general and multidimensional setup. Therefore, these proofs are omitted here. Nevertheless, the approach has two main drawbacks. The first one is model-related: It is assumed that observations are identically distributed, not allowing for subject-specific covariates. The other is proof-related: The imposed regularity assumptions are rather restrictive, for instance, the density of the random effects may not be smooth. The second part of this work addresses these two issues. In subsection \ref{asymptotics_noniid}, we allow the inclusion of deterministic covariates. Moreover, we switch from the standard verification of asymptotic properties of the MLE to a more general strategy, which builds upon the concept of \text{local asymptotic normality} (LAN) of statistical experiments as introduced by  \cite{le2012asymptotic}, and further studied by \cite{ibragimov2013statistical}. It allows statements on asymptotic normality of estimators even when the density functions have some degree of "roughness". The general conditions for consistency and asymptotic normality of the MLE that were formulated by \cite{ibragimov2013statistical} are adapted to our setting. We then discuss intuitive conditions (on, for instance, the $N$-sample Fisher information), which are familiar from regression settings (such as $\frac{1}{N}I_N(\theta)\rightarrow I(\theta)$) and which help to verify the assumptions for MLE asymptotics. We point out the difficulties that arise with these conditions when observations are generated by not identically distributed SDMEMs, and propose a way to remedy these issues. 
Section \ref{simulations} covers simulations for multivariate models that are linear in the fixed and random effects, the latter being Gaussian distributed. In this setting the likelihood is explicitly available. The first example is linear in state  and covariates. More specifically, we consider a multidimensional Ornstein-Uhlenbeck model with one covariate having two levels (e.g., \textit{treatment} and \textit{placebo}). This type of model is  common in, e.g., pharmacokinetics, and is motivated by a recent study on Selenium metabolism in humans \citep{ruse2015absorption}. Moreover, we perform hypothesis testing on the treatment effect and investigate the performance of the Wald test. The second example is the stochastic Fitzhugh-Nagumo model, used to model electrical activity in neurons, which, after parametrization, is still linear in the parameter, but non-linear in the state. We explore the quality of estimation for different sample sizes and sampling frequencies.

\section{Preliminaries}
This section introduces the general statistical setup and derives the likelihood function. \looseness=-1

\subsection{The setting}
We consider $N$ $r$-dimensional stochastic processes $X^i = (X^i_t)_{0\leq t\leq T^i}$ whose dynamics are governed by the stochastic differential equations
\begin{align}\label{SDE}
dX_t^{i}& = F(t,X_t^{i},\mu, \phi^i)dt + \Sigma(t,X_t^{i}) dW_t^i, \quad 0\leq t\leq T^i, \quad X_0^{i}= x_0^i,   \qquad i=1,\ldots,N. 
\end{align}
The $r$-dimensional Wiener processes $W^i = (W^i_t)_{t\geq 0}$ and the $d$-dimensional random vectors $\phi^i$ are defined on some filtered probability space $(\Omega,\mathcal{F},(\mathcal{F}_t)_{t\geq 0}, \mathbb{P})$, which is rich enough to ensure independence of all random objects $W^i,\phi^i, i=1,\ldots,N$. The $d$-dimensional vectors $\phi^i, i=1,\ldots,N$, are the so-called \textit{random effects}. They are assumed to be $\mathcal{F}_0$-measurable and  have a common (usually centered) distribution which is specified by a (parametrized) Lebesgue density $g(\varphi;  \vartheta)d\varphi$. The parameter $\vartheta\in\R^{q-p}$ is unknown as well as the \textit{fixed effect} $\mu\in\R^{p}$. Together, these two quantities are gathered in the parameter $\theta = (\mu, \vartheta)$. This is the object of statistical inference and is assumed to lie in the parameter space $\Theta$, which is a subset of $\mathbb{R}^q$.  
The functions $F:[0,T]\times\R^r\times\R^p\times\R^d\rightarrow\R^r, \Sigma:[0,T]\times\R^r\rightarrow\R^{r\times r}$ with $ T = \max_{1\leq i\leq N} T^i$, are deterministic and known and the initial conditions $x_0^i$ are independent and identically distributed (i.i.d.) $r$-dimensional random vectors.  We assume  that we observe $X^i$  at time points $0 \leq t^i_0 < t^i_{1}<\ldots <t^i_{n_i} = T^i$  and the inference task consists in recovering the "true" underlying $\theta$ based on the observations $X^i_{t_{0}^i},\ldots, X^i_{t_{n}^i}$ of $X^i, i=1,\ldots, N$. To this end, we first suppose to have the entire paths $ (X^i_t)_{0\leq t\leq T^i}$,\\ $ i=1,\ldots,N,$ at our disposal and derive the continuous-time MLE. Later on we will investigate the error that arises when it is approximated by its discrete-time analogue. 

\begin{standAss}
	To assure that the inference problem is well-defined, we assume that the coefficient functions $F, \Sigma$ and the distributions of  initial conditions $x^i_0$ and random effects $\phi^i$ are such that \eqref{SDE} has unique, continuous solutions $X^i, i=1,\ldots,N$, satisfying $\sup_{0\leq t\leq T^i}\mathbb{E}\left(\abs{X^{i}_t}^{k}\right)<\infty$  for all $k\in\mathbb{N}$. If not further specified, it is assumed that $\mu,\varphi$  lie in relevant subspaces of $\R^p$ and $\R^d$, respectively. Natural choices would be the projection of $\Theta$ on the $\mu$-coordinate and the support of the $g(\cdot; \vartheta)$ - or the largest support of those, if $g(\cdot; \vartheta)$ and $g(\cdot; \tilde{\vartheta})$ have different supports for different $\vartheta, \tilde{\vartheta}$. We assume that for all $\mu,\varphi$ (in relevant subspaces) there is a continuous, adapted solution $X^{i,\mu, \varphi}$ to
	\begin{align}\label{SDEvarphi}
	dX^{i,\mu, \varphi}_t &= F(t,X^{i,\mu, \varphi}_t, \mu, \varphi)dt + \Sigma(t,X^{i,\mu, \varphi}_t)dW^i_t, \quad 0\leq t\leq T^i, \quad X_0^{i,\mu,\varphi} = x_0^i,
	\end{align}
	with existing moments of any order (see above). We moreover assume that there are $\mu_0,\varphi_0$ such that for all $\mu,\varphi$ the $X^{i,\mu, \varphi}$ satisfy, with $\Gamma = \Sigma\tt{\Sigma}$, 		  
	$$\mathbb{P}\left(\int_0^{T}\tt{F(s,X_s^{i,\mu, \varphi}, \mu_0, \varphi_0)}\inv{\Gamma(s,X^{i,\mu, \varphi}_s)} F(s, X_s^{i,\mu, \varphi}, \mu_0, \varphi_0) ds <\infty\right) = 1.$$ 
\end{standAss}

\begin{remark}
	Sufficient assumptions for the above are the standard Lipschitz and sublinear growth conditions on $F$  and  $\Sigma$. 
\end{remark}

In all what follows, an integral of a matrix will be understood component-wise, $\abs{\cdot}$ denotes the Euclidean vector norm, $\babs{\cdot}$ a sub-multiplicative matrix norm, $\tt{A}$ the transpose of a matrix $A$ and $\inv{A}$ its inverse. Further,  all statements given for subject $i$ are implied to hold for all $i=1,\ldots,N$. To ease notation, we will write $F_{\mu,\varphi}(s,x)$ instead of $F(s,x,\mu,\varphi)$ and simply $F(s,x)$ for $F_{\mu_0,\varphi_0}(s,x)$.

We denote by $C_{T}$ the space of continuous $\R^r$-valued functions defined on $[0,T]$, which is endowed with the Borel-$\sigma$-algebra $\mathcal{C}_{T}$, the latter being associated with the topology of uniform convergence. A generic vector in $\R^r$ is denoted by $x$, while we use $\mu$ for one such in $\R^p$ and $\varphi$ for one in $\R^d$.  
If $X^{i,\mu, \varphi}$ is the unique continuous solution to \eqref{SDEvarphi}, we denote the measure induced by $X^{i,\mu, \varphi}$ on the space $(C_{T^i},\mathcal{C}_{T^i})$  by $\mathbb{Q}_{\mu, \varphi}^i$ and for $\phi^i\sim g(\cdot; \vartheta)$ and $X^{i}$ as the unique continuous solution to \eqref{SDE}  under $\mathbb{P}_\theta$, we write $\mathbb{Q}^i_\theta$ for the distribution of $X^i$ on $(C_{T^i},\mathcal{C}_{T^i})$. On the product space $\prod_{i=1}^{N} C_{T^i}$ we introduce the product measures $\mathbb{Q}_{N,\mu, \varphi} = \bigotimes_{i=1}^N \mathbb{Q}^i_{\mu, \varphi}$ and $\mathbb{Q}_{N,\theta} = \bigotimes_{i=1}^N \mathbb{Q}^i_\theta$.
Expectations with respect to $\mathbb{Q}^i_{\mu,\varphi}$ and $\mathbb{Q}^i_\theta$ are written as $\mathbb{E}_{\mu, \varphi}$ and $\mathbb{E}_\theta$, respectively, (for convenience, omitting the index $i$) and expectation with respect to the joint distribution of $(\phi^i, X^i)_{1\leq i\leq N}$ on the product space $\prod_{i=1}^{N} \left(\R^d\times C_{T^i}\right)$ will be denoted by $\mathbb{E}_\theta$. 
From now on, we let $\left(\phi^i,\cp^i\right)$ be the canonical process on $\R^d\times C_{T^i}$.

\subsection{Derivation of the likelihood}
One can show that the conditional distribution of $X^i$ conditioned on $\phi^i = \varphi$ coincides with the distribution of $X^{i, \mu,\varphi}$. This, combined with Fubini's theorem, implies that the likelihood of $X^i$, denoted by $p^i(\theta)$, is the integral over the likelihood $q^i(\mu,\varphi)$ of $X^{i,\mu,\varphi}$, weighted by $g(\varphi; \vartheta)$, that is, $p^i(\theta) = \int q^i(\mu,\varphi) g(\varphi; \vartheta) d\varphi$. The following theorem, which is a standard result from the theory on SDEs, specifies the likelihood $q^i(\mu,\varphi)$.

\begin{thm}[Conditional likelihood]\label{conditional_likelihood}
	The distribution $\mathbb{Q}^i_{\mu, \varphi}$  is dominated by\\ $\nu^i:=\mathbb{Q}^i_{\mu_0, \varphi_0}$ and the Radon-Nikodym derivative is a.s. given by
	\begin{align*}
	q^i(\mu,\varphi; X^i) &=
	\text{exp}\left( \int_{0}^{T^i}\tt{\left[F_{\mu,\varphi}(s, \cp_s^i) - F(s,\cp_s^i)\right]}\inv{\Gamma(s,\cp^i_s)}d\cp^i_s \right. \\
	& \qquad- \left.\frac{1}{2}\int_0^{T^i}\tt{\left[F_{\mu,\varphi}(s, \cp_s^i) - F(s,\cp_s^i)\right]}\inv{\Gamma(s,\cp^i_s)}\left[F_{\mu,\varphi}(s, \cp_s^i) + F(s,\cp_s^i)\right]ds\right).
	\end{align*}
\end{thm}

\begin{standAss}{}
	To assure that  $p^i$ is well-defined and measurable, the integrand 
	${(\mu, \varphi, \vartheta, x)\mapsto q^i(\mu,\varphi; x)g(\varphi; \vartheta)}$ should be measurable (w.r.t. the product Borel-$\sigma$-algebra). Therefore, we assume that the model is such that $g$ and $q^i$ are both product-measurable.  From above, we see that $q^i$ is measurable in the $x$-component, such that $q^i(\mu,\varphi; x)$ is surely product-measurable if it is continuous in its remaining two components. 
	For instance, a sufficient condition for continuity in $\varphi$ is 
	that for any $\mu$  there  is $\kappa = \kappa(\mu)>0$ such that for all $\varphi_0, \varphi,  x$ (in relevant sets)  and $0\leq s\leq T^i$, one has 
	$\abs{\tt{\left[F_{\mu, \varphi}(s, x) - F_{\mu, \varphi_0}(s,x)\right]}\inv{\Gamma(s,x)}} \leq K(1+\abs{x}^\kappa)\abs{\varphi-\varphi_0}.$
	This, together with a linear growth condition on $\Sigma$, implies that there is a $\tilde{\kappa}$ such that 
		\begin{align*}
		\abs{(\tt{\left[F_{\mu, \varphi}- F_{\mu, \varphi_0}\right]}\inv{\Gamma}\left[F_{\mu, \varphi} - F_{\mu, \varphi_0}\right])(s,x)} \leq C(1+\abs{x}^{\tilde{\kappa}})\abs{\varphi-\varphi_0}^2. 
		\end{align*}
	One can then apply Kolmogorov's continuity criterion, which will  yield the  continuity (rather, existence of an in $\varphi$ continuous version) of $q^i$ in $\varphi$. 
\end{standAss}

The  likelihood of the sample ${X}^{(N)} = (X^1,\ldots, X^N)$ is now an immediate consequence of Fubini's theorem.

\begin{thm} [Unconditional likelihood]\label{final_density}
	The distribution $\mathbb{Q}^i_\theta$ admits the $\nu^i$-density 
	$p^i(\theta):=p^i(\theta; X^i) = \int_{\R^d}q^i(\mu,\varphi)\cdot g(\varphi;\vartheta)\, d\varphi$
	and the corresponding product measure $\mathbb{Q}_{N,\theta}$  has the $\nu_N$-density 
	$p_N(\theta):= p_N(\theta; {X}^{(N)})= \prod_{i=1}^{N}p^i(\theta)$ (with $\nu_N =\bigotimes_{i=1}^N\nu^i$).
\end{thm}
We remark that the absolute continuity of $\mathbb{Q}^i_{\theta}$ w.r.t. $\nu^i$ implies that all $\nu^i$-a.s. statements made in the sequel also hold $\mathbb{Q}^i_\theta$-a.s., for all $\theta\in\Theta$.\\

\section{Asymptotic results for the MLE}
The present section deals with asymptotic properties of the MLE (consistency, asymptotic normality and time discretization) and is divided into two parts. The first part, which assumes identically distributed observations, is in spirit  close to the work of \cite{delattre2013maximum} and extends their results to a multidimensional state process with time-inhomogeneous dynamics. Given a drift function that is linear in the fixed and random effects (but possibly nonlinear in the state variable), we derive results on consistency and asymptotic normality of the MLE, following the traditional road of proof. Moreover, we leave the  theoretical setting of continuous-time observations and switch to the practical situation, in which data are only available at discrete time points. We bound the discretization bias that arises when the continuous-time statistics are replaced by their discrete-time analogues.
The second part is more general in two aspects: On the one hand, we allow the drift to be subject-specific by inclusion of covariate information and on the other hand, we suggest an alternative road to verification of consistency and asymptotic normality of the MLE, which poses less regularity assumptions and is based on the LAN property of the statistical model(s).

\subsection{Independent and identically distributed observations}\label{asymptotics_iid}
In this first subsection, we state results on consistency and asymptotic normality of the MLE when observations are independent and identically distributed. The proofs, which closely follow those in \cite{delattre2013maximum}, while, however, getting more tedious to work out due to the multidimensional setup, are omitted here, but are available upon request. The drift function in \eqref{SDE} is assumed to be linear in the fixed and random effects and random effects have a centered $d$-dimensional Gaussian distribution with unknown covariance matrix $\Omega$, such that the likelihood is explicitly available. More specifically, we consider the \\ $r$-dimensional processes $X^i$, whose dynamics are given by
\begin{align}\label{linear_SDE}
dX_t^i = \left[A(t,X_t^i) + B(t,X_t^i)\mu + C(t, X_t^i) \phi^i\right]dt + \Sigma(t,X_t^i) dW_t^i, \; X_0^i = x_0^i,\, 0\leq t\leq T,
\end{align}
and the parameter to be estimated based on the (continuous-time) observations $(X^i_t)_{0\leq t\leq T}$,  $ i=1,\ldots,N,$ is $\theta = (\mu,\Omega)$. The parameter space $\Theta$ is a bounded subset of $\R^p\times\mathfrak{S}_d(\mathbb{R})$, where $\mathfrak{S}_d(\mathbb{R})$ is the set of symmetric, positive definite $(d\times d)$-matrices. The conditional likelihood of subject $i$ is (cf. Theorem \ref{conditional_likelihood})
$q^i\left(\mu,\varphi\right)  = \e{ \tt{\mu} U_{1i} - \frac{1}{2}\tt{\mu}V_{1i}\mu + \tt{\varphi}U_{2i} - \frac{1}{2}\tt{\varphi}V_{2i}\varphi - \tt{\varphi}S_i\mu}$
with the sufficient statistics 
\small{
\begin{align*}
U_{1i}&= \int_{0}^{T}(\tt{B}\inv{\Gamma})(s,\cp^i_s)\left[d\cp^i_s -A(s,\cp_s^i)ds\right],   \quad\text{}\quad
V_{1i} = \int_0^{T}(\tt{B}\inv{\Gamma}B)(s,\cp^i_s)ds,\\
U_{2i}&= \int_{0}^{T}(\tt{C}\inv{\Gamma})(s,\cp^i_s)\left[d\cp^i_s -A(s,\cp_s^i)ds\right],   \quad\text{}\quad
V_{2i} = \int_0^{T}(\tt{C}\inv{\Gamma}C)(s,\cp^i_s)ds,\\
S_i   & = \int_0^{T}(\tt{C}\inv{\Gamma}B)(s,\cp^i_s)ds.
\end{align*}}
According to Theorem \ref{final_density}, the $N$-sample log-likelihood then turns into the explicit expression  \\$l_N(\theta) = \log\left(\prod_{i=1}^{N} p^i(\theta) \right)$, where
{\small
\begin{align*}
p^i(\theta)
&= \frac{1}{\sqrt{\text{det}(I + V_{2i}\Omega)}} \exp\left(\left[\tt{U}_{1i} - \tt{U}_{2i}R^i(\Omega) S_i\right] \mu -\frac{1}{2}\tt{\mu}\left[V_{1i} - \tt{S}_iR^i(\Omega) S_i\right] \mu + \frac{1}{2}\tt{U}_{2i}R^i(\Omega)U_{2i} \right)
\end{align*} }
and with $R^i(\Omega) = \inv{(V_{2i}+\inv{\Omega})}$.
We set $G^i(\Omega) =   \inv{\left(I + V_{2i}\Omega\right)} V_{2i}$ and  assuming that $V_{1i},V_{2i}$ are strictly positive definite, we can write 
$P^i(\theta) = G^i(\Omega)\inv{V}_{2i}(U_{2i}-S_i\mu)$, and the (vectorized) $N$-sample score function is given by
$S_N(\theta) = \sum_{i=1}^{N}S^i(\theta) = \left[ \frac{d}{d\mu}l_N(\theta), \frac{d}{d\Omega}l_N(\Omega)'\right] $, with
\begin{align*}
\frac{d}{d\mu}l_N(\theta)  &= \sum_{i=1}^{N}\left[\tt{U}_{1i}-\tt{U}_{2i}R^i(\Omega)S_i\right] - \tt{\mu} \sum_{i=1}^{N}\left[V_{1i} - \tt{S}_i R^i(\Omega)S_i\right],\\
\frac{d}{d\Omega}l_N(\theta)  &= \frac{1}{2}\sum_{i=1}^{N} \left[ -G^i(\Omega) + P^i(\theta) \tt{P^i(\theta)}\right].
\end{align*}
The MLE  $\hat{\theta}_N = (\hat{\mu}_N, \hat{\Omega}_N)$ solves the equations
\begin{equation}
\begin{aligned}
\hat{\mu}_N  &=  \inv{\left[\sum_{i=1}^{N}\left[V_{1i} - \tt{S}_i R^i(\hat{\Omega})S_i\right]\right]}\left[\sum_{i=1}^{N}\left[U_{1i}-\tt{S}_i R_i(\hat{\Omega})U_{2i}\right]\right] \\
 \sum_{i=1}^{N}&\left[-G^i(\hat{\Omega}_N) +  P^i(\hat{\theta}_N) \tt{P^i(\hat{\theta}_N)}\right] = 0.
\end{aligned}
\end{equation}

Note that the likelihood $p^i$ is explicit even if the fixed effect enters the drift nonlinearly. However, only a linear fixed effect $\mu$ leads to an explicit expression for its ML estimator $\hat{\mu}$. If one wants to impose a random variation on all components of $\mu$, one simply sets $B(t,x) = C(t,x)$. In that case, $U_{1i} = U_{2i}=:U_i$ and $V_{1i} = V_{2i} = S_i =:V_i$. The conditional likelihood simplifies to $q^i\left(\mu,\varphi\right)  = \e{ \tt{(\mu+\varphi)} U_{i} - \frac{1}{2}\tt{(\mu+\varphi)}V_{i}(\mu+\varphi)}$ and the unconditional likelihood to 
\begin{align}\label{linear_like}
p^i(\theta)& = \frac{1}{\sqrt{\text{det}(I + V_i\Omega)}} \exp\left(-\frac{1}{2}\tt{(\mu-\inv{V}_i U_i)}G^i(\Omega) (\mu - \inv{V}_i U_i)\right) \exp\left(\frac{1}{2}\tt{U}_i \inv{V}_i U_i\right).
\end{align}
With  $\gamma^i(\theta) = G^i(\Omega)(\inv{V}_iU_i-\mu)$ (and clearly, $\gamma^i(\theta) = P^i(\theta)$), the MLE  $\hat{\theta}_N = (\hat{\mu}_N, \hat{\Omega}_N)$ then solves
\begin{equation}
\begin{aligned}\label{MLE_system}
\hat{\mu}_N  &=  \inv{\left[\sum_{i=1}^{N}  G^i(\hat{\Omega}_N)\right]}\left[\sum_{i=1}^{N} \inv{(I + V_i\hat{\Omega}_N)}U_{i}\right] \\ \sum_{i=1}^{N}&\left[-G^i(\hat{\Omega}_N) +  \gamma^i(\hat{\theta}_N) \tt{\gamma^i(\hat{\theta}_N)}\right] = 0.
\end{aligned}
\end{equation} 
To simplify the subsequent outline, we impose random effects on all components of $\mu$, such that the dynamics are
$dX_t^i = \left[A(t,X_t^i) + C(t, X_t^i) (\mu+\phi^i)\right]dt + \Sigma(t,X_t^i) dW_t^i$  and 
$p^i(\theta)$ is given by \eqref{linear_like}.

\begin{lemma}[Moment properties]\label{finite_moments}
	For all $\theta \in \Theta$ the following statements hold.
	\begin{enumerate}		
		\item[(i)]  For  all $u\in\R^d$,	$\mathbb{E}_\theta\left(\exp\left(\tt{u}\inv{(I + V_i\Omega )}U_i\right)\right) <\infty.$
		\item[(ii)] $ \mathbb{E}_\theta\left(\gamma^i(\theta)\right) = 0, \quad	\mathbb{E}^i_\theta\left(\gamma^i(\theta)\tt{\gamma^i(\theta)} -G^i(\Omega)\right) = 0.$
	\end{enumerate}
\end{lemma}

The asymptotic normality of the normalized score function is an immediate consequence of the law of large numbers  together with the standard multivariate central limit theorem.

\begin{thm}[Asymptotic normality of the normalized score function]\label{convScore}
	For all $\theta\in\Theta$, under $\mathbb{Q}_\theta$ and as $N$ tends to infinity, the (vectorized  version of) the normalized score function $\sqrt{N}S_N(\theta)$ 
	converges in distribution to $\mathcal{N}(0,\mathcal{I}(\theta))$, where $\mathcal{I}(\theta)$ is the covariance matrix of 
	$\text{vec}(S^i(\theta)),$ with $S^i(\theta) =  \left[
	\gamma^i(\theta), \; \frac{1}{2}\left(\gamma^i(\theta)\tt{\gamma_i(\theta)} - G^i(\Omega)\right)
	\right]$.
\end{thm}

\begin{standAss}\label{assD}
	\mbox{}
	\begin{enumerate}
		\item 
		The function $x\mapsto (C'\inv{\Gamma}C)(s,x):\R^r \rightarrow \R^{d\times d}$ is not constant and under $\mathbb{Q}^i_{\varphi_0}$ the $\R^{d\times (d+1)}$-valued random variable 	$(U_i,V_i)$	admits a continuous density function  (w.r.t. the Lebesgue measure), which is positive on an open ball of $\Theta\subset\R^d\times\R^{d\times d}$. 
		\item $\Theta$ is convex and compact. In particular, we assume that there are positive constants $M_1, M_{2,1}$, $M_{2,2}$ such that for all $\theta\in\Theta$: $\abs{\mu}\leq M_1$ and $M_{2,1}\leq \babs{\Omega}\leq M_{2,2}$. 
		\item 
		The true value $\theta_0$ belongs to int$(\Theta)$ and 	the matrix $\mathcal{I}(\theta_0)$ is invertible. 
	\end{enumerate}
\end{standAss}

\begin{thm}[Continuity of KL information and uniqueness of its minimum]\label{KL}
\mbox{}\\
	Let $K(\mathbb{Q}^i_{\theta_0}, \mathbb{Q}^i_\theta) = \mathbb{E}_{\theta_0}(\log p^i(\theta_0)-\log p^i(\theta))$ be the Kullback-Leibler information of $\mathbb{Q}^i_{\theta_0}$ w.r.t. $\mathbb{Q}^i_{\theta}$. 
	Then the function $\theta\mapsto K(\mathbb{Q}^i_{\theta_0}, \mathbb{Q}^i_\theta)$ is continuous  and has a unique minimum at $\theta = \theta_0$.
\end{thm}

\begin{thm}[Weak consistency and asymptotic normality of the MLE]\label{consistency}
	Let $\hat{\theta}_N$ be an ML estimator defined as any solution of $l_N(\hat{\theta}_N) = \sup_{\theta\in\Theta}l_N(\theta)$. Then, as $N\rightarrow\infty$, $\hat{\theta}_N$ converges to $\theta_0$ in $\mathbb{Q}_{\theta_0}$- probability, and $\sqrt{N}(\hat{\theta}_N-\theta_0)\rightarrow\mathcal{N}\left(0, \mathcal{I}^{-1}(\theta_0)\right)$  under $\mathbb{Q}_{\theta_0}$.
\end{thm}

\subsubsection{Discrete data}
So far, we have assumed that we observe the entire paths $(X^i_t)_{0\leq t\leq T}$ of the processes generated by \eqref{SDE}. 
This is a severe restriction as in practice, observations are only available at discrete time points $t_{0},\ldots,t_{n}$. A natural approach is to replace the  continuous-time integrals in $q^i(\theta)$ by discrete-time approximations and to derive an approximate MLE based on the resulting approximate likelihood. For instance, the stochastic integral term in $q^i(\theta)$, which is an expression of the form $\int_{t_k}^{t_{k+1}} h(s,X^i_s)dX^i_s$, may be replaced by a first-order approximation 
$\int_{t_k}^{t_{k+1}} h(s,X^i_s) dX^i_s \approx h(t_k, X^i_k) \Delta X^i_k$
or, by a higher-order approximation using Ito's formula,   giving
$\int_{t_k}^{t_{k+1}}  h(s,X^i_s) dX^i_s \approx H(t_{k+1},X^i_{k+1}) - H(t_k, X^i_k) -\frac{\Delta t}{2}\int_{t_k}^{t_{k+1}}\sum_{j,l=1}^{r}(H_{x_j,x_l}\Sigma_j\Sigma_l')(t_k,X_k^i),$
where $h(t,x) = \nabla_x H(t,x)$. Note that this requires $h$ to be of gradient-type, i.e., it requires the existence of a differentiable function $H$ such that $h$ can be obtained as $h(t,x) = \nabla_x H(t,x)$. A higher-order approximation scheme is preferable, if the time step is not sufficiently small (non-high-frequency data) and/or the dynamics are highly non-linear. 
In the linear model \eqref{linear_SDE}, the  first-order approximation of the continuous-time likelihood (which breaks down to the discretization of the sufficient statistics $U_i,V_i$) corresponds to the exact likelihood of its Euler scheme approximation. In particular, if we assume for simplicity that we observe all individuals at time points $t_i = \frac{i}{n}$ and denote by $U_i^n, V^n_i$ the first-order discrete-time approximations to the continuous-time statistics $U_i,V_i$, one has the following result: 

\begin{thm}
	Assume model \eqref{linear_SDE} and suppose that $A(t,x), (\tt{C}\inv{\Gamma}C)(t,x)$ and $(\tt{C}\inv{\Gamma})(t,x)$ are globally Lipschitz-continuous in $t$ and $x$ and that in addition to $A(t,x), C(t,x)$ and $\Sigma(t,x)$ also $(\tt{C}\inv{\Gamma})(t,x)$ is of sublinear growth in $x$, uniformly in $t$. 
	Then for any $p\geq 1$, the error behaves like\\ $\mathbb{E}_{\theta}\left(\babs{V_{i} - V_{i}^n}^p + \abs{U_{i}-U_{i}^n}^p \right) =  O(n^{-p/2})$.
\end{thm}

\subsection{Non-homogeneous observations and covariates}\label{asymptotics_noniid}
In this section, we consider the asymptotic behavior of the MLE when the observations \\$X^i = (X^i_t)_{0\leq t\leq T^i}$, $i=1,\ldots,N,$  are independent between subjects $i$, but not necessarily identically distributed. This  occurs, for instance, if the drift contains subject-specific covariate information $D^i$ and these covariates are not i.i.d. If they are assumed to be deterministic,  as in standard regression,  the drift function $F$ varies to a certain degree across subjects, $F^i(t,x,\mu,\phi^i) = F(t,x,D^i, \mu,\phi^i)$. As in the i.i.d. case, one would naturally wonder, under which conditions on the degree of variation among the $F^i$ the derived MLEs still satisfy standard asymptotic results and, equally important, how to verify conditions that assure a regular asymptotic behavior. \\
Results on asymptotic normality of MLEs for independent, not identically distributed (i.n.i.d.) random variables are well-known \citep{bradley1962asymptotic,hoadley1971asymptotic}. They commonly built upon regularity conditions  on the density functions $p^i$, such as third-order differentiability and boundedness of the derivatives, to ensure that integration and differentiation can be interchanged  (as in the i.i.d. case, see also subsection \ref{asymptotics_iid}). To achieve a limiting behavior when the observations do not share a common distribution, the variation across these non-homogeneous distributions has to be  controlled. This is usually achieved by, on the one hand, imposing that the family of score functions $\{S^i(\theta); i\in\mathbb{N}\}$  satisfies the Lindeberg condition (a condition that bounds the variation of each $S^i(\theta)$ in relation to the total variation of the $N$-sample score function $\sum_{i=1}^{N}S^i(\theta)$). On the other hand, by requiring that the sample average $\frac{1}{N}\sum_{i=1}^{N}I^i(\theta)$ of the Fisher information matrices $I^i(\theta)$  converges to a positive definite limiting matrix $I(\theta)$. Under these conditions, the Lindeberg-Feller central limit theorem assures that the scaled $N$-sample score function is asymptotically $\mathcal{N}(0,I(\theta))$ distributed and a Taylor expansion gives the asymptotic normality of the MLE \citep{bradley1962asymptotic, hoadley1971asymptotic, gabbay2011philosophy}.\\
The regularity conditions imposed on the densities as mentioned above are as standard as  restrictive. If, for instance, random effects are supposed to have a double exponential distribution, i.e., a distribution whose density is not differentiable at $x=\theta$, those regularity conditions can not be met.  The Laplace density is, however, "almost" regular. In fact, it satisfies a particular type of first-order differentiability and can perfectly be treated by a less standard, but more general road to verification of consistency and asymptotic normality. It dispenses with the previously mentioned strong regularity conditions on the density functions and instead  builds upon $L_2$-differentiability and the LAN property of a sequence of statistical models \citep{le2012asymptotic, ibragimov2013statistical}.

\subsubsection{The convergence of the averaged Fisher informations }
When studying the asymptotic behavior of the MLE in the setting of independent, but not identically distributed observations, a common - and intuitive -  assumption is to require that the sample average $\frac{1}{N}\sum_{i=1}^{N}I^i(\theta) = \frac{1}{N}I_N(\theta)$ of the individual Fisher information matrices converges to a deterministic, symmetric, positive definite (SPD) limit matrix $I(\theta)$ as the sample size grows to infinity, (see, e.g.,  \citet[condition N7]{bradley1962asymptotic}, or \citet[equation (13)]{hoadley1971asymptotic}). This not only simplifies verification of the assumptions  considerably, it is also  natural when compared to the i.i.d. case, where $I_N(\theta) = NI(\theta)$ and $\left[ \frac{1}{N}I_N(\theta)\right]^{-1} = I(\theta)^{-1}$ is the  asymptotic variance of the scaled MLE $\sqrt{N}\hat{\theta}_N$. However, it would be convenient to break the requirement down to the level of the model structure. If, for instance, the only structural difference between the distributions of the $X^i$ is caused by inclusion of covariates, it is natural to ask whether one can formulate conditions on the average behavior of the covariates, as it is done in standard linear regression. This, however, is not possible for SDMEM, not even if we assume the simplest case where the drift function $F$ is linear in state, covariates, fixed and random effects and if the latter are Gaussian distributed with known covariance matrix. And here is why. 
Assume a standard linear regression model $y_N = \mathcal{X}_N\mu + \epsilon_N$ with $N$ observations collected in the response vector $y_N$, deterministic $N\times p$ design matrix $\mathcal{X}_N$ containing the covariate information for all subjects, unknown parameter vector $\mu$ and $N$-dimensional vector $\epsilon_N$ of uncorrelated noise. The Fisher information matrix is given by $I_N(\mu) = \mathcal{X}_N'\mathcal{X}_N$  and the standard assumption is $\frac{1}{N}(\mathcal{X}_N' \mathcal{X}_N)\rightarrow V$ for some SPD matrix $V$. The matrix $I_N(\mu)$ has elements $\sum_{i=1}^{N}x^i_k x^i_l$, therefore the convergence requirement translates to assuming that second order sample averages of the covariates converge.  If the linear model additionally includes random effects $\phi\sim\mathcal{N}(0,\Omega)$, that is, $y_N = \mathcal{X}_N\mu + \mathcal{Z}_N \phi + \epsilon_N$ and $\mathcal{Z}_N$ is a deterministic $N\times d$ design matrix, a standard assumption (see, e.g., \cite{pinheiro2006mixed}) for asymptotic normality is (among others) the convergence of $\frac{1}{N}\mathcal{X}_N' R_N(\Omega)^{-1}\mathcal{X}_N$ to a SPD $V$, where  $R_N(\Omega) = I+\mathcal{Z}_N\Omega \mathcal{Z}_N'$ is the covariance matrix of $y_N$. So also for linear mixed effects models one can break the convergence of the average Fisher information down to conditions on a second-order average behavior of the covariates. In the SDE case the situation is more difficult. In fact, assuming moment statistics of the covariates to converge is \textit{not} enough to ensure convergence of the Fisher information matrix, which we illustrate now in the simplest possible example that includes covariates. 
Assume $r=1$ and $d=2$ and consider the dynamics $dX_t^i = \left[X_t^i(\mu^1 + \phi^{i,1}) + D_t^i(\mu^2 + \phi^{i,2})\right] dt + dW_t^i$ for $0\leq t\leq T$ with $X_0^i = x_0$.  The vector $\mu = (\mu^1,\mu^2)'$ is the unknown fixed effect and the $\phi^i = (\phi^{i,1},\phi^{i,2})'$, $i=1,\ldots,N$, are i.i.d. two-dimensional  random effects with $\mathcal{N}(0,\Omega)$ distribution. Assume that the covariance  matrix $\Omega$ is known, such that $\theta = \mu$ is the only unknown parameter.  This setup is a special case of the example in subsection \ref{affine_with_covariates}.  The sufficient statistics $U_i,V_i$ are given by \looseness=-1
\begin{align*}
U_i =\begin{pmatrix}\int_0^TX_t^i dX_t^i\\ \int_0^T D_t^i dX_t^i\end{pmatrix} \quad\text{ and }\quad V_i = \begin{pmatrix}\int_0^{T^i} (X_t^i)^2dt
& \int_0^{T^i} X_t^iD_t^i dt\\ \int_0^{T^i} X_t^i D_t^idt & \int_0^{T^i} (D_t^i)^2 dt\end{pmatrix}.
\end{align*} 
The Fisher information is by definition $I^i(\mu) =  \mathbb{E}_\mu\left( -\frac{d}{d\mu}S^i(\mu) \right) = \mathbb{E}_\mu\left( f(U_i, V_i) \right)$, where the function $f$ is the negative second derivative of the log-likelihood function. Since the log-likelihood is quadratic in $\mu$, $f$ will  in fact not depend on the parameter (since $\Omega$ is known).  We  immediately conclude from the expression of $p^i$ in eq. \eqref{linear_like} that $f(U_i,V_i) = G^i(\Omega) =  (I+V_i\Omega)^{-1} V_i $, such that $I^i(\mu) =  \mathbb{E}_\mu\left(G^i(\Omega)\right)$. The $2\times 2$ matrix $G^i(\Omega)=\inv{\left(I + V_{i}\Omega\right)} V_{i}$ is, however, a non-linear function of $V_i$ and thus finding an explicit expression for $I^i(\mu)$ is generally impossible -  even in the simple linear case, where $X^i$ is nothing but a Gaussian process. 
For comparison, in the linear mixed effects model, the log-likelihood for observation $y^i$ with covariate vectors $x^i, z^i$ is proportional to $-\frac{1}{2} (y^i-(x^i)'\mu)'R^i(\Omega)(y^i-(x^i)'\mu)$, with $R^i(\Omega) = (I+z_i' \Omega z_i)^{-1}$ as inverse covariance matrix of $y^i$. The Fisher information is $\mathbb{E}_\mu\left(f(x^i, z^i)\right) $ with $f(x^i, z^i) = x^iR^i(\Omega)(x^i)'$. The crucial difference, as compared to the SDE case, is that the matrix $R_i(\Omega)$ is deterministic. This renders calculation of the expectation unnecessary, such that $I^i(\mu)= \mathbb{E}_\mu\left(f(x^i, z^i)\right)= f(x^i, z^i)$. Therefore, requiring convergence of $\frac{1}{N}\sum_{i=1}^{N}I^i(\theta)$ is nothing but  asking for a limiting behavior of covariate averages $\frac{1}{N}\sum_{i=1}^{N} f(x^i,z^i)$. This is particularly attractive as one can often design the experiment in such a way that the required limiting behavior holds. 
In the SDE case, however, it will not - not even in the simple linear case  - be possible to break the condition $\frac{1}{N}I_N(\theta)\rightarrow I(\theta)$ down to the level of covariates, by requiring that an expression of the form $\frac{1}{N}\sum_{i=1}^{N} f(D^i,\theta)$, with $f$ being some suitable function, converges. 
Therefore, it will generally not be possible to determine from an analytical expression of $I_N(\theta)$, whether the condition $\frac{1}{N}I_N(\theta)\rightarrow I(\theta)$ holds! Of course, this is not the end of the day, as the direct way via specific expressions for $I_N(\theta)$ is not the only possible road to show convergence. Averages of the form $\frac{1}{N}\sum_{i=1}^{N}a_i$ converge, for instance, if the sequence $\{a_i\}$ converges to a limit $a$ as $i\rightarrow\infty$.  In this spirit, 
an alternative way would be to, heuristically speaking, assume that everything  which is deterministic and individual-specific, or random but not with the same distribution across all individuals, converges as $i$ goes to infinity to a limit (for instance,  $x_0^i \rightarrow x_0, T^i\rightarrow T, D^i\rightarrow D$) in a suitable sense. Such an assumption corresponds to requiring that for large $i$ the observations are, in fact, identically distributed. Exemplified in the linear example, one could proceed as follows. The Fisher information based on observation $X^i$ is  $I^i(\theta) = \mathbb{E}_\theta(f(U_i, V_i, \theta))$ and $U_i = U(x_0^i, T^i, D^i),V_i = V(x_0^i, D^i, T^i)$ with suitable functions $f,U,V$. If one can show that $U, V$ and $f$ are  continuous in all arguments, a.s. convergence of $x_0^i,T^i, D^i$ to limits $x_0, T, D$ (as $i\rightarrow\infty$)  implies convergence of $f(U^i,V^i,\theta)$. If the family $\{ f(U^i,V^i, \theta); \; i\in \mathbb{N}  \}$ is uniformly integrable, the a.s. convergence implies the convergence of moments and thus $I^i(\theta)\rightarrow I(\theta)$ and therefore also the average converges, i.e.  $\frac{1}{N}I_N(\theta) = \frac{1}{N}\sum_{i=1}^{N}I^i(\theta)\rightarrow I(\theta)$. Note, however, that due to the nonlinear dependence of the function $V$ on $D^i$, convergence of covariate averages of, for instance, the form $\frac{1}{N}\sum_{i=1}^{N} D^i_t \rightarrow D_t$ is not enough to ensure that the averaged Fisher informations converge - even if the other quantities $x_0^i, T^i$ are the same for all individuals.

\subsubsection{Asymptotics of the MLE with generalized conditions}
A framework that also captures less regular models is provided by \cite{ibragimov2013statistical} and will here be adapted to the present setting. Those results that are not included in \cite{ibragimov2013statistical} are  adaptations of ones therein and proofs will  be omitted. We make the following assumptions.

\begin{standAss}{}
	\mbox{}
	\begin{itemize}
		\item  $\theta\mapsto p^i(\theta)$  is $\nu^i$-a.s. continuous  and  $\theta\mapsto \sqrt{p^i(\theta)}$ is $L_2(\nu^i)$-differentiable, i.e., $p^i(\theta)$ is \textit{Hellinger differentiable} with $L_2(\nu^i)$-derivative  $\psi^i(\theta)$ (a row vector).  That is, for each $\theta$, \\$\int \tabs{\psi^i(\theta;x)}^2d\nu^i(x)<\infty$,   ${\lim_{\tabs{h}\rightarrow0}\tabs{h}^{-2} \int \tabs{\sqrt{p^i(\theta+h; x)} - \sqrt{p^i(\theta; x)} - \psi^i(\theta; x) h}^2d\nu^i(x) =0} $.\looseness=-1 
		\item  $\psi^i(\theta)$ is continuous in $L_2(\nu^i)$. 
        Consequently, the matrix  $I^i(\theta) = 4\int \psi^i(\theta; x)'\psi^i(\theta; x)d\nu^i(x)$ exists, is continuous and the $N$-sample Fisher information matrix can be defined as $I_N(\theta) = \sum_{i=1}^{N} I^i(\theta)$. 
		\item  $0< \inf_{\theta\in \Theta} \babs{\frac{1}{N}I_N(\theta)} \leq \sup_{\theta\in\Theta}\babs{\frac{1}{N}I_N(\theta)}<\infty$.
		\item There is a SPD matrix $I(\theta)$ such that $\lim_{N\rightarrow\infty}\sup_{\theta\in K}\babs{\frac{1}{N}I_N(\theta)-I(\theta)} \rightarrow 0$ and \\$\lim_{N\rightarrow\infty}\sup_{\theta\in K}\babs{\left(\frac{1}{N}I_N(\theta)\right)^{-1/2} - I(\theta)^{-1/2}} = 0$.
	\end{itemize}
\end{standAss}
Assuming that the (norm of the) Fisher information matrix grows beyond bounds corresponds to the requirement of infinite flow of information, which is naturally connected to the consistency of estimators. The $L_2$-differentiability is neither a stronger nor weaker concept than standard (point-wise) differentiability. One may think of the relation between the two differentiability concepts as of the one between $L_2$-convergence and almost sure convergence -  without further assumptions, in general none of them implies the other, but under certain conditions, the limits are identical. Of course, if $p^i$ is $L_2$-differentiable and differentiable in the ordinary sense, then $\psi^i(\theta; x) = \frac{d}{d\theta}\left[ p^i(\theta; x)^{1/2}\right]$. \\
Analogously to the traditional setting, we  call $S^i(\theta) = 2 p^i(\theta)^{-1/2} \psi^i(\theta)$  the score function of sample $i$ and set $S_N(\theta) = \sum_{i=1}^{N}S^i(\theta)$ for the $N$-sample score function. A result familiar from traditional theory is that the score function is centered, which under the above conditions also holds true here, \cite[p. 115]{ibragimov2013statistical}. 
The likelihood ratio process (random field), which will be defined on the local parameter space $\Theta_{N,\theta} = \{h\in\R^q:\, \theta + I_N(\theta)^{-1/2}h\in\Theta\}$, is denoted by $L_{N,\theta}(h) = p_N(\theta + I_N(\theta)^{-1/2}h) / p_N(\theta)$.

\begin{remark}
		\mbox{}
	\begin{enumerate}
		\item Sufficient conditions for the (a.s.) continuity of  $p^i(\theta) = \int q^i(\mu,\varphi) g(\varphi; \vartheta) d\varphi$ in $\theta$  are continuity of $\mu\mapsto q^i(\mu,\varphi)$ and $\vartheta\mapsto g(\varphi; \vartheta)$, together with the existence of a dominating, integrable function, $q^i(\mu,\varphi) g(\varphi; \vartheta) \leq H(\varphi)$.
		If the density $g$ of the random effects is  assumed to be Gaussian, it is naturally continuous in $\vartheta$ (provided the variance parameter is bounded away from zero). For the continuity of $q^i$ in $\mu$, we remark the following: 
		Suppose  $F$ is continuous and assume for simplicity $\Sigma(t,x)\equiv I$ is the identity matrix. If $F$ is uniformly continuous in $\mu$ (for instance differentiable with bounded derivative/Jacobian), then $\mu\mapsto \int_0^{T^i}F(s,x(s),\mu,\varphi)'F(s,x(s),\mu,\varphi)ds$ is continuous. If $F$ moreover has the property 
		$\abs{F(s, x, \mu, \varphi) - F(s,x,\mu_0, \varphi)} \leq K(1+\abs{x}^\kappa)\abs{\mu-\mu_0}$ for some $\kappa>0$, Kolmogorov's continuity criterion yields the  continuity (rather, existence of an in $\mu$ continuous version) of $q^i$ in $\mu$.
		\item 	Suppose $\theta\mapsto \sqrt{p^i(\theta)}$ is continuously differentiable. Then (since $p^i>0$) the quantity $\tilde{S}^i(\theta) := 2p^i(\theta)^{-1/2}\frac{d}{d\theta}p^i(\theta)$ is well-defined. If  the expression $\tilde{I}^i(\theta) = \mathbb{E}_\theta(\tilde{S}^i(\theta)\tilde{S}^i(\theta)')$ is finite and moreover continuous, then $\theta\mapsto \sqrt{p^i(\theta)}$ is $L_2$-differentiable \cite[Lemma 7.6]{van2000asymptotic},  the $L_2$-derivative coincides with the point-wise derivative and $\tilde{S}^i(\theta) = S^i(\theta), \tilde{I}^i(\theta) = I^i(\theta)$. 
	\end{enumerate}
\end{remark}

\subsubsection{General results on consistency and asymptotic normality}
In this part, we give conditions  on the asymptotic behavior of the MLE in our present framework. For simplicity, $\Theta\subseteq\R^q$ is assumed to be open, bounded and convex and in all what follows, $K\subset\Theta$ is a (fixed) compact subset. Whenever we write $\theta_N$, we mean that it is of the form $\theta_N = \theta + I_N(\theta)^{-1/2}h$ for $\theta\in K$ and $h\in\Theta_{N,\theta}$.

\begin{thm}[Consistency]\label{I51}
	The MLE is uniformly on $K$ consistent, if
	\begin{enumerate}
		\item[(A.1)] There is a constant $m>q$ such that	 $\sup_{\theta\in K}\mathbb{E}_\theta\left(\abs{S_N(\theta)}^m\right)<\infty$.
		\item[(A.2)] There is a positive constant $a(K)$  such that for (sufficiently large $N$ and) all $\theta\in K$ (and all $h\in\Theta_{N,\theta}$)
		$H^2_i(\theta, \theta_N) \geq a(K)\frac{\abs{\theta_N-\theta}^2}{1+\abs{\theta_N-\theta}^2},$
		where $H^2_i(\theta_1,\theta_2) = \int \left(\sqrt{p^i(\theta_1)} - \sqrt{p^i(\theta_2)}\right)^2 d\nu^i$ is the squared Hellinger distance between $\mathbb{Q}^i_{\theta_1}$ and $\mathbb{Q}^i_{\theta_2}$.
	\end{enumerate}
\end{thm}
\begin{proof}
	(A.1) is an extension of Lemma III.3.2. in \cite{ibragimov2013statistical} to the setting of non-homogeneous observations and (A.2) is adapted from \cite[Lemma I.5.3]{ibragimov2013statistical}.
\end{proof}

\begin{remark}
	If the dimension of the parameter set is 1, the first condition above can be replaced by a sub-quadratic growth condition on the Hellinger distance (for i.i.d. observations, see \citet[Theorem I.5.3]{ibragimov2013statistical}). In that case, one  can instead require that $H^2(\theta_1,\theta_2)\leq A \tabs{\theta_2-\theta_1}^2$, such that for one-dimensional parameter sets,  consistency here reduces to $H^2(\theta_1,\theta_2)$ behaving asymptotically as $\tabs{\theta_2-\theta_1}^2$. 
\end{remark}

The following theorem implies the so-called uniform asymptotic normality of the model, which in turn gives rise to the asymptotic normality of the MLE  (cf. Theorems II.6.2. and III.1.1 in \cite{ibragimov2013statistical}). 

\begin{thm}[Asymptotic normality]\label{III11}  
	Assume (A.1) and (A.2) from Theorem \ref{I51} hold. If additionally 
	\begin{enumerate}
		\item[(B.1)] $\{S^i(\theta), i=1,\ldots, N\}$ satisfy the Lyapunov condition uniformly in $K$, i.e.
		there is $\delta >0$ such that
		$	\lim_{N\rightarrow\infty}\, \sup_{\theta \in K}\; \sum_{i=1}^{N}\mathbb{E}_\theta\left(\abs{I_N(\theta)^{-1/2}S^i(\theta)}^{2+\delta}\right) = 0.$			
		\item[(B.2)] $\forall R>0:$ 
		$\lim_{N\rightarrow\infty}\; \sup_{\theta\in K} \;\sup_{\tabs{h}<R}\quad \sum_{i=1}^N\int_{C_{T^i}}\left( \left[ \psi^i(\theta_N) -  \psi^i(\theta)\right] I_N(\theta)^{-1/2} h\right)^2 d\nu^i = 0,$
	\end{enumerate}
	are satisfied, $\{\hat{\theta}_N\}_{N\in\mathbb{N}}$ is uniformly in $K$ consistent, asymptotically Gaussian distributed with parameters  $(\theta, I_N(\theta)^{-1})$ and all moments of $\{I_N(\theta)^{1/2}(\hat{\theta}_N - \theta)\}_{N\in\mathbb{N}}$ converge uniformly in $K$ to the corresponding moments of the $\mathcal{N}(0,I)$ distribution.
\end{thm}

Condition (B.1) can be generalized to the Lindeberg condition.  If the densities $\sqrt{p^i(\theta)}$ are twice continuously differentiable with second derivative $J^i(\theta)$,  (B.2) can be replaced by requiring that\\
$\lim_{N\rightarrow\infty}\;\sup_{\theta\in K}\;\sup_{\tabs{h}\leq R}\quad\babs{I_N(\theta)^{-1/2}}^4 \sum_{i=1}^{N}\int_{C_{T^i}}\babs{J^i(\theta_N)}^2d\nu^i = 0. $
In the general setting, the $p^i$  are  not explicitly available.  One can, however, formulate more general conditions on the drift function $F$ and on the random effects density $g$ such that differentiability of $\log p^i(\theta) = \log\left( \int q^i(\mu,\varphi) g(\varphi; \vartheta) d\varphi\right)$ is guaranteed, by assuring that differentiation can be passed under the integral sign.
Sufficient conditions for the differentiability of $\log p^i(\theta)$ with respect to $\mu$ would, e.g., include differentiability of $q^i$ w.r.t. $\mu$ and a uniform in $\mu$  domination of \\ $\frac{d}{d\mu}  q^i(\mu,\varphi)  d\varphi (\int q^i(\mu,\varphi) g(\varphi; \vartheta) d\varphi)^{-1}$. Explicitly formulating these conditions is not very illustrative. Instead, it is recommended to check suitable conditions in the specific application at hand. One particular case in which the $p^i(\theta)$ are explicitly available is the case of linear Gaussian random effects, which will be considered in more detail below.\looseness=-1

\subsubsection{Affine Gaussian fixed and random effects and inclusion of covariates}\label{affine_with_covariates}
We revisit the example model \eqref{linear_SDE}, but now include for each subject $i$ a covariate information $D^i$, which is a known and deterministic function $D^i:[0,T^i]\rightarrow \R^{s}$. We let $B(t,X_t^i)=C(t,X_t^i)$ in \eqref{linear_SDE} and enrich the function $C$ by the covariate, $C(t,X^i_t,D^i)$. This model, being linear in state, covariate information,  fixed effect $\mu\in\R^d$ and in the $d$-dimensional random effects $\phi^i$, is the simplest non-trivial models with covariates. We assume that $\Theta$ is a bounded subset of $\R^d\times\mathfrak{S}_d(\mathbb{R})$.
The likelihood and score function are as in \eqref{linear_like} and \eqref{MLE_system}, respectively, the only difference being a possible subject-specific observation horizon $T^i$ and that the covariate information now enters the sufficient statistics $U_{i} = \int_{0}^{T^i}\tt{C(s,\cp_s^i, D_s^i)}\inv{\Gamma(s,X_s^i)}\left[d\cp^i_s -A(s,\cp_s^i)ds\right]$ and $  V_{i} = \int_0^{T^i}\tt{C(s,\cp^i_s, D_s^i)}\inv{\Gamma(s,X_s^i)}C(s,\cp^i_s, D_s^i)ds$. Again, we assume that $V_i$ is invertible. 
It is clear that the model is more regular than actually required and we include this example, where we verify the conditions of Theorem \ref{III11},  merely for illustration purposes. It will also be revisited in the subsequent section on simulations, where we investigate parameter estimation (and hypothesis testing) for different sample sizes and sampling frequencies.   

The set $K\subset\Theta$ is compact, so there are positive constants $A_K, B_K, C_K$ such that $\abs{\mu}\leq A_K, B_K\leq \babs{\Omega}\leq C_K$. 
One can show that $\babs{G^i(\Omega)}\leq \babs{\inv{\Omega}}$, which gives the upper bound 
$\tabs{S^i(\theta)} \leq  \left(\tabs{\gamma^i(\theta)} + \babs{\inv{\Omega} 
} + \tabs{\gamma^i(\theta)}^2\right)$. Moreover, the moment-generating function $\Phi_{\theta,\gamma^i(\theta)}(a)$ of $\gamma^i(\theta)$ can be bounded from above by   $\e{\frac{1}{2}\tt{a}\inv{\Omega} a}$, for $a\in\mathbb{R}^d$. This can be used to find that $\mathbb{E}_\theta\left(\tabs{\gamma_i(\theta)}^m\right) \leq C_1$ for some constant $C_1$ that may depend on $K,d,m$.  Therefore, there is another constant $C_2$, which may depend on $K,d,m,N,$ such that $\mathbb{E}_\theta\left( \tabs{S_N(\theta)}^m\right) \leq C_2$, proving (A.1). 
To verify (A.2), note that the regularity of $p_N(\theta)$ and its derivatives implies that 
{\small\begin{align*}
H^2(\theta, \theta_N) 
& = \int \left[-\psi_N(\theta)(\theta_N-\theta) +  \left(\sqrt{p_N(\theta_N)} - \sqrt{p_N(\theta)}\right) + \psi_N(\theta)(\theta_N-\theta) \right]^2d\nu\\
& = \int \left[-\psi_N(\theta)(\theta_N-\theta)\right]^2 d\nu + o(\tabs{\theta_N-\theta}^2) \\
& = (\theta_N-\theta)' I_N(\theta) (\theta_N-\theta) + o(\tabs{\theta_N-\theta}^2) - 2 O(\tabs{\theta_N-\theta}^2) o(\tabs{\theta_N-\theta}^2)\\
&\geq \tabs{(\theta_N-\theta)}^2 \lambda_{N,\text{min}}(\theta)+ o(\tabs{\theta_N-\theta}^2). 
\end{align*}}
where $\lambda_{N,\text{min}}(\theta)$ denotes the smallest eigenvalue of $I_N(\theta)$. 
Therefore, for $N$ sufficiently large, there is a constant $A_K$ such that $H^2(\theta, \theta_N) \geq A_K \tabs{(\theta_N-\theta)}^2$. Since $\Theta$ is bounded, we even have $\tabs{(\theta_N-\theta)}^2\geq C\frac{\tabs{(\theta_N-\theta)}^2}{1+\tabs{(\theta_N-\theta)}^2}$ for some positive constant $C$, which shows that (A.2) holds.  The Lyapunov condition (B.1) follows in a straightforward way. According to the above,  $\mathbb{E}_\theta\left(\tabs{S^i(\theta)}^3\right)\leq C$ for some $C$ and therefore
\begin{align*}
\sup_{\theta\in K}\sum_{i=1}^{N}\mathbb{E}_\theta\left(\tabs{I_N(\theta)^{-1/2}S^i(\theta)}^{3}\right) 
& \leq N^{-3/2}\sup_{\theta\in K} \babs{\sqrt{N}I_N(\theta)^{-1/2}-I(\theta)^{-1/2}} \sum_{i=1}^{N}\mathbb{E}_\theta\left(\abs{S^i(\theta)}^3\right) \\
&\quad + N^{-3/2}\sup_{\theta\in K} \babs{I(\theta)^{-1/2}}\sum_{i=1}^{N}\mathbb{E}_\theta\left(\tabs{S^i(\theta)}^{3}\right) \\
& \leq C N^{-1/2} \left[\sup_{\theta\in K} \babs{\sqrt{N}I_N(\theta)^{-1/2}-I(\theta)^{-1/2}}  + \sup_{\theta\in K} \babs{I(\theta)^{-1/2}}\right], 
\end{align*}
which converges to 0 as $N\rightarrow\infty$. To verify  (B.2), we show that 
\begin{align}\label{two_conditions}
\sup_{\tabs{h}\leq R}\; \frac{1}{N}\left[\frac{1}{N} \sum_{i=1}^{N} \mathbb{E}_{\nu^i}\left( \bbabs{J^i(\theta_N) - J^i(\theta)}^2\right)\right] \;\text{ and }\;
\frac{1}{N}\left[\frac{1}{N} \sum_{i=1}^{N} \mathbb{E}_{\nu^i}\left( \bbabs{J^i(\theta)}^2\right)\right]
\end{align}  
converge to 0 uniformly in $K$. As $J^i(\theta)$ is continuous, it is uniformly continuous on compacta, such that for all $i\in\mathbb{N}$, $ a_{i,N}= \sup_{\tabs{h}\leq R}\bbabs{J^i(\theta_N)-J^i(\theta)}$ converges a.s. to 0 as $N\rightarrow\infty$. One can show that $a_{i,N} \leq A^i(\theta,R)$ and $\mathbb{E}_{\nu^i}\left( A^i(\theta,R)^2\right)\leq D_K$. Dominated convergence therefore implies $\mathbb{E}_\theta(a_{i,N})\rightarrow 0$,  and the uniform (in $i$) bound $D_K$ implies uniform in $K$ convergence of the left term in \eqref{two_conditions} to 0. For the right hand side term in \eqref{two_conditions} we note that $\mathbb{E}_{\nu^i}\left(\babs{J^i(\theta)}^2\right) 
\leq \mathbb{E}_\theta\left( \babs{\frac{d}{d\theta} S^i(\theta)}^2 \right) + \mathbb{E}_\theta\left(  \babs{S^i(\theta)'S^i(\theta)}^2\right) < C_K$, where $C_K$ is a constant that only depends on $K$. We conclude uniform in $\theta\in K$ convergence of the right hand side term in \eqref{two_conditions} to 0. The right hand side follows similarly.

\subsection{Hypothesis testing}\label{hypothesis}
It is commonly of interest to the researcher to test whether an applied treatment  has a significant effect on the treated subjects, i.e., to test whether an underlying treatment effect $\beta$, a $s$-dimensional subparameter of the fixed effect $\mu$, $1\leq s\leq p$,  is  significantly different from 0. The asymptotic normality of the MLE in this model lends itself naturally to the application of Wald tests, which can be used to investigate two-sided null hypotheses such as $H_0: \beta = 0$ (no treatment effect) or more generally, any $k$-dimensional, $1\leq k\leq s$, linear null hypothesis $H_0: L\beta = \eta_0$, where $L$ is a $5\times k$ matrix of rank $k$ which specifies the linear hypotheses of interest and $\eta_0\in \R^k$. The Wald test statistic is 
$\hat{W}_N = \left(  L\hat{\beta}_N - \eta_0  \right)' \left(L\hat{V}_NL'\right)^{-1} \left(  L\hat{\beta}_N - \eta_0  \right)',$
which is under the null hypothesis asymptotically $\chi^2$-distributed with $k$ degrees of freedom \citep{lehmann2006testing}. Here,  $\hat{\beta}_N$ is the MLE of $\beta$ and ${\hat{V}_N = \widehat{\mathbb{C}\text{ov}}(\hat{\beta}_N)}$ denotes its estimated variance-covariance matrix of $\hat{\beta}_N$.

\section{Simulations}\label{simulations}
\subsection{Linear transfer model}
The first example, which is inspired from a study on the selenomethionine metabolism in humans \citep{ruse2015absorption}, is a 5-dimensional linear transfer model, which finds applicability in various fields, especially in modeling population flows or in pharmacokinetics. A component in the model's state vector can be viewed to represent the concentration of a substance in a certain compartment and the model describes the (linear) flow between compartments. We consider a basic cascade-shaped transfer structure as illustrated in Figure \ref{fig:cascade}.  When observing $N$ subjects, each of them following the linear transfer model in Figure \ref{fig:cascade}, it is often reasonable to assume that the transfer rates are subject-specific.  We moreover assume that we are given covariate information on subject $i$ in form of a (deterministic) variable $D^i\in\{0,1\}$. It encodes the affinity of subject $i$ to one of two possible study groups, such as \textit{placebo} and \textit{treatment}.  Consider the model (for simplicity assuming unit diffusion) $dX_t^i = F(t,X_t^i, D^i, \mu,\phi^i) dt + dW_t^i$,  $0\leq t\leq T, X_0^i = 0,$ where $\mu' = (\alpha,\beta)$ is the fixed parameter and the drift function has the specific form $F(t,X^i_t,D^i,\mu,\phi^i) = -G(\alpha + \phi^i) X_t + D^i\beta$, $i=1,\ldots,N$, with rate matrix
\begin{align*}
G(\alpha)= \begin{pmatrix} 
\alpha_1 & 0 & 0 & 0 & -\alpha_5 \\
-\alpha_1     & \alpha_2 & 0 & 0 & 0 \\
0 & -\alpha_2 & \alpha_3+\alpha_6 &  0 & 0 \\
0 & 0 & -\alpha_3 & \alpha_4 & 0 \\
0 & 0 & 0 & -\alpha_4 & \alpha_5 \end{pmatrix}.
\end{align*}
The (unknown) fixed effect $\mu$ has the 6-dimensional component $\alpha$, which is shared across both groups (\textit{placebo} and \textit{treatment}) and an additional 5-dimensional component $\beta$, which describes the effect of the covariate (treatment effect) on the subjects' dynamics. We let  $\beta' = (1,2,3,1,-2)$. The random effects $\phi^i$ are  i.i.d. $\mathcal{N}(0,\Omega)$-distributed and the covariance matrix $\Omega$ is unknown.  With $\alpha' = (\alpha_1,\ldots,\alpha_6) = (2, 4, 3, 2, 1,1) $, all eigenvalues of $G(\alpha)$ have positive real parts, implying that the model has a stationary solution. The processes $X^i$ for individuals without treatment, i.e.  $D^i = 0$, are  (on average) mean-reverting to 0, and the processes belonging to individuals in the treatment group, $D^i = 1$, are mean-reverting to the long-term mean $(G(\alpha+\phi^i))^{-1}\beta$. For our choice of parameters, conditional on $\phi^i = 0$, this long-term mean is $(7.50, 4.25, 5.00, 8.00, 14.00)'$, see also Figure \ref{fig:mOU_trace_plots}. 
The covariance matrix $\Omega$  is taken to be a diagonal matrix with  entries diag$(\Omega)=\left(0.5^2, 1^2, 1^2, 0.5^2, 0.3^2, 0.3^2\right)$. The observation horizon $T$ is fixed to $T=15$. A trajectory of $(X^1_t,\ldots, X_t^N)_{0\leq t\leq T}$ is simulated with the Euler-Maruyama scheme with simulation step size $\delta = 10^{-4}$. Figure \ref{fig:mOU_trace_plots} shows four realized (5-dimensional) trajectories of the process $X^i$. The upper two panels show trajectories for $D^i = 0$ and the lower two correspond to trajectories with $D^i = 1$. \looseness=-1

\subsubsection{Parameter estimation}
For parameter estimation, the simulated trajectories are thinned by a factor $b$ (taking only every $b$-th observation). To explore the expected time-discretization bias of the estimators, we repeated estimation for different thinning factors, $b\in\{10, 100, 1000\}$, which results in sampling intervals  $\Delta t = \delta\cdot b =  0.001, 0.01, 0.1$. To also investigate the estimation performance  as a function of  sample size, we performed estimation on trajectories  $(X^1,\ldots, X^N)$, for sample sizes $N = 20$, $N=50$ and $N = 100$.  Estimation for all considered $(\Delta t, N)$-combinations was repeated on $M=500$ simulated data sets. Tables \ref{tab:Ex1_1} and \ref{tab:Ex1_2} report the sample estimates of relative biases and root mean squared errors (RMSE) of the fixed effects and of the variances of the random effects. The relative bias of $\hat{\alpha}_j$ is computed as $ \frac{1}{M}\sum_{m=1}^{M}\frac{\hat{\alpha}_j^{(m)} - \alpha_j}{\alpha_j}$ and the RMSE as $\left(\frac{1}{M}\sum_{m=1}^{M}(\hat{\alpha}_j^{(m)} - \alpha_j)^2\right)^{1/2}, j=1,\ldots, 6$,  and with an analogous definition for the other parameters.   Table \ref{tab:Ex1_1} shows estimation results for a fixed sample size of $N=50$, and different values of sampling intervals $\Delta t$, while results in Table \ref{tab:Ex1_2} are computed for a fixed sampling interval $\Delta t = 0.001$ and different values of sample size $N$. In each table, the first six rows correspond to estimated biases and RMSEs of the shared fixed effects $\alpha_j$, $j=1,\ldots, 6$. The subsequent five rows show the estimated biases and RMSEs of the treatment effects $\beta_j$, $j=1,\ldots,5$ and the last six rows correspond to the estimated biases and RMSEs of the diagonal elements of $\Omega$ (i.e., the variances of the random effects). 
The estimation is very accurate already at sample sizes as small as $N=20$, when the data is sampled at high frequency (here $1/0.001$), see Table \ref{tab:Ex1_2}. For a moderate sampling frequency of $1/0.01$, the results in the middle part of Table \ref{tab:Ex1_1} reveal that estimates of the  fixed effects $\alpha, \beta$ are on average biased by only about 1-2\% (of the true parameter value), which is still very accurate. The variances of the random effects are estimated with an average bias of 5-9\% for $N=50$ and $\Delta t = 0.01$. When the observations are sampled at low frequency $1/0.1$, estimation gets unreliable. The bias due to the time-discretization of the continuous-time estimator is very pronounced, with values of up to 25\% for the fixed effects and up to almost 50\% for the variances of the random effects. The RMSEs rise - as compared to a 10 times higher frequency -  by more than 100\%. If only low-frequency data is available, caution is recommended and estimation should only be done on a data set that has been enlarged by imputing data in between the observation time points.\looseness=-1

\begin{figure}[h!]
	\begin{center}
		\includegraphics[ scale=0.75]{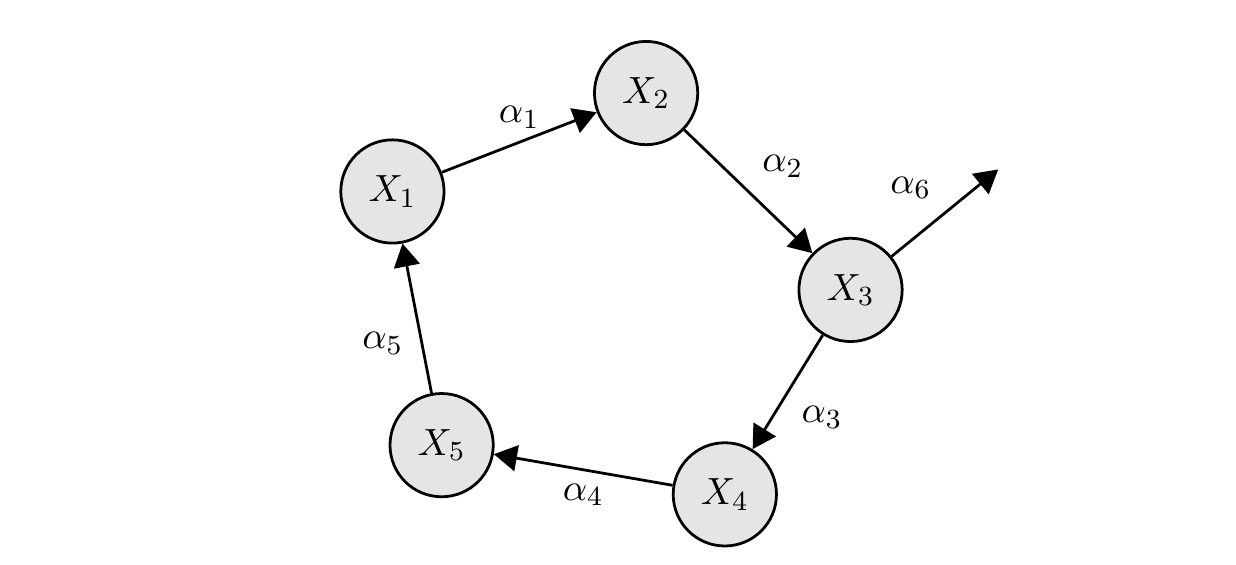}
		\caption{	\label{fig:cascade}\footnotesize Illustration of the $5$-dimensional linear transfer model used in the first simulation example. The state $X_i = (X_{i,t})_{0\leq t\leq T}$ gives  the concentration (over time) of a substance in compartment $i, i=1,\ldots,5$. The $\alpha_i, i=1,\ldots,6$ are the unknown rates of flow between corresponding compartments (for $\alpha_6$ the outflow of the system). }
	\end{center}
\end{figure}
\vspace{-0.15in}
\begin{figure}[h!]
	\begin{center}
	\includegraphics[ scale=0.35]{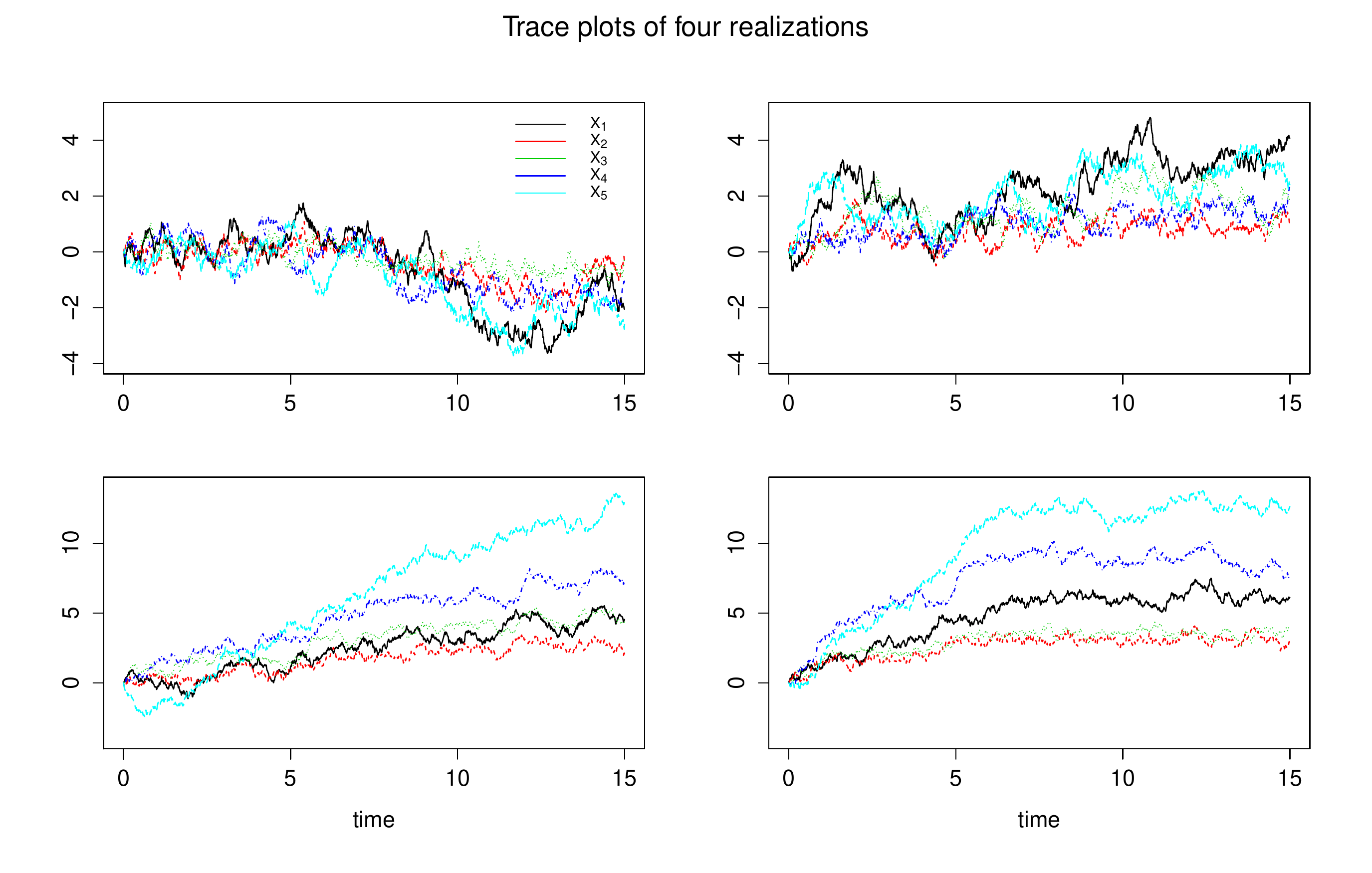}
	\end{center}
	\caption{	\label{fig:mOU_trace_plots}\footnotesize Linear transfer model: Four realizations of the 5-dimensional state process $(X^i_t)_{0\leq t\leq T}$. The upper two panels show realizations when the covariate is 0 ("reference group") and the lower two panels display trajectories for $D^i= 1$ ("treatment group"). Note the clearly visible difference in the long-term means between the two groups.}
\end{figure}

{\footnotesize
{\renewcommand{\baselinestretch}{0.6} 
\begin{table}
	\caption{	\label{tab:Ex1_1}\footnotesize Linear transfer model. Shown are estimated relative bias and RMSE of $\hat{\alpha}, \hat{\beta},$ and diag$\left(\hat{\Omega}\right)$. 
	The sample size is fixed to $N=50$, but different sampling intervals are considered ($\Delta t = 0.001, 0.01, 0.1$). 
	For each value of $\Delta t$, the estimation was repeated on $M=500$ generated data sets.}
\centering
\begin{tabular}{r r@{\hskip 1cm} rr r@{\hskip 1cm} rr r@{\hskip 1cm} rr}
\\
\hline\\
& &\multicolumn{2}{c}{$\Delta t=0.001$} & & \multicolumn{2}{c}{$\Delta t = 0.01$} & & \multicolumn{2}{c}{$\Delta t = 0.1$}\\[0.2cm]
\cline{3-4}   \cline{6-7}  \cline{9-10}\\
true value && rel. bias & RMSE && rel. bias & RMSE && rel. bias & RMSE \\[0.2cm]
\hline\\
2.00 &  & 0.001 & 0.079 &  & -0.018 & 0.086 &  & -0.182 & 0.369 \\ 
4.00 &  & -0.002 & 0.149 &  & -0.024 & 0.172 &  & -0.204 & 0.824 \\ 
$\alpha\quad$3.00 &  & 0.001 & 0.163 &  & -0.021 & 0.170 &  & -0.203 & 0.624 \\ 
2.00 &  & -0.001 & 0.083 &  & -0.017 & 0.088 &  & -0.162 & 0.332 \\ 
1.00 &  & 0.001 & 0.047 &  & -0.016 & 0.049 &  & -0.159 & 0.164 \\ 
1.00 &  & 0.002 & 0.091 &  & -0.008 & 0.091 &  & -0.082 & 0.119 \\[0.2cm]  
1.00 &  & -0.002 & 0.099 &  & -0.020 & 0.099 &  & -0.166 & 0.186 \\ 
2.00 &  & -0.002 & 0.114 &  & -0.024 & 0.121 &  & -0.198 & 0.408 \\ 
$\beta\quad$3.00 &  & 0.002 & 0.152 &  & -0.010 & 0.152 &  & -0.116 & 0.373 \\ 
1.00 &  & -0.001 & 0.148 &  & 0.014 & 0.146 &  & 0.140 & 0.188 \\ 
-2.00 &  & 0.002 & 0.124 &  & -0.024 & 0.131 &  & -0.255 & 0.522 \\[0.2cm]  
0.25 &  & -0.037 & 0.062 &  & -0.079 & 0.062 &  & -0.399 & 0.108 \\ 
1.00 &  & -0.035 & 0.208 &  & -0.095 & 0.216 &  & -0.483 & 0.498 \\ 
diag$(\Omega)\quad$1.00 &  & -0.035 & 0.215 &  & -0.085 & 0.219 &  & -0.426 & 0.445 \\ 
0.25 &  & -0.026 & 0.061 &  & -0.065 & 0.060 &  & -0.352 & 0.097 \\ 
0.09 &  & -0.009 & 0.022 &  & -0.047 & 0.021 &  & -0.333 & 0.034 \\ 
0.09 &  & -0.040 & 0.036 &  & -0.065 & 0.035 &  & -0.213 & 0.036  \\[0.2cm]  
		\hline
	\end{tabular}
\end{table}}}

{\footnotesize
	{\renewcommand{\baselinestretch}{0.6} 
	\begin{table}
	\caption{	\label{tab:Ex1_2}\footnotesize Linear transfer model. Shown are estimated relative bias and RMSE of  $\hat{\alpha}, \hat{\beta},$ and diag$\left(\hat{\Omega}\right)$. The sampling interval is fixed to $\Delta t=0.001$, but different sample sizes are considered ($N = 20, 50, 100$). For each value of $N$, the estimation was repeated on $M=500$ generated data sets.}
		\begin{tabular}{r r@{\hskip 1cm} rr r@{\hskip 1cm} rr r@{\hskip 1cm} rr}
			\\
			\hline\\
			 & &\multicolumn{2}{c}{$N=20$} & & \multicolumn{2}{c}{$N=50$} & & \multicolumn{2}{c}{$N=100$}\\[0.2cm]
			\cline{3-4}   \cline{6-7}  \cline{9-10}\\
			true value && rel. bias & RMSE && rel. bias & RMSE && rel. bias & RMSE \\[0.2cm]
			\hline\\
2.00 &  & 0.003 & 0.116 &  & 0.001 & 0.079 &  & -0.001 & 0.058 \\ 
4.00 &  & 0.001 & 0.232 &  & -0.002 & 0.149 &  & 0.001 & 0.114 \\ 
$\alpha\quad$3.00 &  & 0.003 & 0.253 &  & 0.001 & 0.163 &  & -0.001 & 0.106 \\ 
2.00 &  & -0.003 & 0.126 &  & -0.001 & 0.083 &  & -0.000 & 0.052 \\ 
1.00 &  & 0.003 & 0.074 &  & 0.001 & 0.047 &  & -0.003 & 0.031 \\ 
1.00 &  & -0.003 & 0.146 &  & 0.002 & 0.091 &  & 0.000 & 0.068 \\[0.2cm]
1.00 &  & 0.000 & 0.157 &  & -0.002 & 0.099 &  & 0.004 & 0.073 \\ 
2.00 &  & -0.001 & 0.174 &  & -0.002 & 0.114 &  & 0.002 & 0.075 \\ 
$\beta\quad$3.00 &  & 0.002 & 0.233 &  & 0.002 & 0.152 &  & 0.000 & 0.102 \\ 
1.00 &  & 0.010 & 0.231 &  & -0.001 & 0.148 &  & -0.002 & 0.102 \\ 
-2.00 &  & 0.006 & 0.203 &  & 0.002 & 0.124 &  & -0.000 & 0.087 \\[0.2cm] 
0.25 &  & -0.091 & 0.093 &  & -0.037 & 0.062 &  & -0.014 & 0.043 \\ 
1.00 &  & -0.046 & 0.355 &  & -0.035 & 0.208 &  & -0.020 & 0.162 \\ 
diag$(\Omega)\quad$1.00 &  & -0.073 & 0.343 &  & -0.035 & 0.215 &  & -0.017 & 0.163 \\ 
0.25 &  & -0.035 & 0.097 &  & -0.026 & 0.061 &  & -0.016 & 0.039 \\ 
0.09 &  & -0.045 & 0.035 &  & -0.009 & 0.022 &  & -0.021 & 0.015 \\ 
0.09 &  & -0.181 & 0.055 &  & -0.040 & 0.036 &  & -0.020 & 0.027 \\[0.2cm]  
			\hline
		\end{tabular}
	\end{table}}}

\subsubsection{Hypothesis testing}
A natural step  is to test whether $\beta$, or a subparameter, is significantly different from 0. We estimate the false-positive rate of the Wald test  (see subsection \ref{hypothesis}) in this model and investigate the test's power under different "true" (non-zero) treatment effects. The estimated variance-covariance matrix ${\hat{V}_N = \widehat{\mathbb{C}\text{ov}}(\hat{\beta}_N)}$ of $\hat{\beta}_N$ is obtained from $M=500$ (separately) computed MLEs $\hat{\beta}_N^{(m)}, m=1,\ldots,M$, where underlying data sets have been simulated under the true hypothesis (under $H_0$ for estimation of the false positive rate and under $H_1$ for power estimation).  
Tables \ref{tab:Ex1_1} and \ref{tab:Ex1_2} show that the estimation was accurate for high- and medium-frequency observations. Diagnostic plots (not shown here) reveal that the asymptotic distribution of the MLE is close to normal already for $N=20$ subjects, such that even for a rather small data set and a medium sampling frequency, test results can be considered  sufficiently reliable. The choice $(N,\Delta t) = (20, 0.01)$ provides a simulation setting that is sufficiently reliable, but at the same time not trivial and will challenge the hypothesis test, in particular for small treatment effects. The estimated false positive rate (based on $M$ under $H_0$ generated data sets) is $0.074$, revealing a slightly liberal finite-sample test behavior. The power of detecting a treatment effect (rejecting $H_0: \beta = 0$) was computed for different "true" values of $\beta$. For $\beta = (1,2,3,1,-2)'$ (values as in the estimation part above), the estimated power was 1. This comes to no surprise as the  long-term mean $(7.5, 4.25, 5, 8, 14)'$ of the state process in the treatment group is considerably different from the zero long-term mean of the control group. The power, estimated to 0.956,  was still convincing for a much smaller treatment effect $\beta = (0.1, 0.2, 0.3, 0.1, -0.2)'$, which gives a long-term mean of $(0.75, 0.425, 0.5, 0.8,  1.4)'$. This is especially impressive as the state process' standard deviation (from its long-term mean 0) under $H_0$ is about $(0.66, 0.49, 0.59, 0.72, 1.21)'$. More challenging is the rejection of $H_0$ when the treatment  has a small effect on, e.g.,  only one coordinate, $\beta = (0.1, 0, 0, 0, 0)'$. In this case (long-term mean $(0.2, 0.1, 0.1, 0.15, 0.3)'$), and for such a small sample size the chance of rejecting $H_0$ is as small as 16\% and it is thus hardly possible to detect a difference between groups. However, while being only slightly conservative, the asymptotic Wald test is able to detect a treatment effect for a rather small data set, even if it causes only a little change of the long-term mean as compared to the standard deviation of the process. 

\subsection{Fitzhugh-Nagumo model}
The deterministic Fitzhugh-Nagumo (FHN) model  \citep{fitzhugh1955mathematical, nagumo1962active}  is a two-dimensional approximation of the well-known four-dimensional Hodgkin-Huxley neuronal model \citep{hodgkin1952quantitative}  and is typically applied to model the regenerative firing mechanism in an excitable neuron. Neural firing is a complex interplay of numerous cell processes and to account for various unexplained noise sources, a stochastic FHN model can be considered \citep{jensen2012markov}, 
\begin{eqnarray}\label{FHN}
\begin{aligned}
dY_{t} &= \frac{1}{\varepsilon}\left(Y_{t} - Y_{t}^3 - Z_{t} + s\right)dt + \sigma_1 dW_{1,t},\\
dZ_{t} &= \left(\gamma Y_{t} - Z_{t} + \eta\right)dt + \sigma_2 dW_{2,t}.
\end{aligned}
\end{eqnarray}
The  variable $Y$ represents the membrane potential of a neuron, while the $Z$ coordinate  represents  the  recovery. The time scale separation $\varepsilon$ is commonly $\ll 1$, such that $Y$ lives on a much faster time scale than $Z$. The variable  $s$ is the input current. If $\gamma>1$, the system has exactly one fixed point, which may be stable or unstable, depending on the specific parameter values. 
Under the reparametrization
$\mu= \tt{(1/\varepsilon, s/\varepsilon, \gamma, \eta)}$, \eqref{FHN} may be written as in \eqref{linear_SDE}.
We assume to  study a collection of $N$ excitable neurons and model their membrane potentials $Y^i$ via   $dX_t^i = A(X_t^i) + C(X_t^i)(\mu + \phi^i) dt + \Sigma \, dW^i_t$, $0\leq t\leq T$, where $X^i = (Y^i, Z^i)'$ and the $\phi^i$ are the i.i.d. $\mathcal{N}(0,\Omega)$-distributed random effects. Observe that despite being nonlinear in the state variable, the model equations are linear in the random effects and therefore an explicit likelihood is available. We assume here that both coordinates of $X^i$ are observed. For all simulations, we let $\sigma_1 = 0.5, \sigma_2 = 0.3$ (assumed known), we fix $T = 20$ and choose the values of the unknown parameters as $\varepsilon = 0.1$,  $s = 0.5$, $\gamma = 1.5$ and $\eta = 1.2$. With this choice of $\eta$ the fixed point of the deterministic FHN system is stable, but small noise levels will suffice to induce large excursions through state space (spikes). 
The covariance matrix of the random effects is fixed as $\Omega = \text{diag}\left(1.5^2, 1^2, 0.2^2, 0.2^2\right)$. The simulation settings are as in the previous example: We simulate each trajectory $(X_t^1, \ldots, X^N_t)_{0\leq t\leq T}$ with the Euler-Maruyama scheme and a simulation time step of $\delta = 10^{-4}$. The estimation is  carried out on the thinned trajectory. We conduct estimation for different values of the sample size $N\in\{20, 50, 100\}$ to investigate the finite sample behavior. To illustrate how the discrete-time bias evolves, we repeat estimation  for thinning factors $b =10, 100, 1000$, which results in sampling intervals of $\Delta t = \delta\cdot b = 0.001, 0.01, 0.1$, respectively (note that the observation horizon is always fixed to $T=20$).  For all combinations of $N$ and $\Delta t$, the estimation is repeated on $M=500$ generated data sets. 
Figure \ref{fig:trace_plots_FHN} shows example trace plots of four realizations, which illustrate the possible qualitatively different behaviors of the state process, depending on the realized values of the random effects. Table \ref{tab:Ex2_1} shows, similar to the previous simulation example on the linear transfer model, the bias and the RMSE of the estimates, where estimation was based on samples with fixed sample size $N=50$, but repeated for different sampling intervals $\Delta t$. The estimation was done under the reparametrization $\mu = (1/\varepsilon, s/\varepsilon, \gamma,\eta)$. Estimates for the parameter $\varepsilon$ and  $s$ on the original scale are obtained by transformation. The upper six rows show estimated bias and RSME for the fixed effects (on the original and on the transformed $\mu$-scale) and the subsequent four rows correspond to results for the estimation of the diagonal of $\Omega$. Despite the non-linearity (in the state) of the model, implying violation of the absolute standard assumptions on the diffusion drift for regularity of the model, the parameter estimation for high-frequency data and moderate ($N=50$) sample size is very convincing (Table \ref{tab:Ex1_1}, first two columns), while still being satisfactory for observations sampled at medium frequency (middle two columns). If observations are sampled at low frequency (last two columns in Table \ref{tab:Ex2_1}), the bias for the estimation of $s, \gamma,\eta$ is still rather low (with 1\%, 8\% and 5\% bias, respectively, as compared to the true parameter value). The estimation of $\varepsilon$ is, however, highly biased. The variances of the random effects are all estimated with an error of about 21-28\%, except for the variance of the random effect that adds to $\varepsilon$, which has an error of as high as 69\%.  This comes to no surprise, since non-linearity in the state requires denser observations. For $\varepsilon$ we estimate the inverse of a small number, making the estimator unstable. Figure \ref{fig:FHN_relErr} illustrates the distribution of the relative bias of the $\nu$-scale estimates (i.e. of the bias divided by the true parameter value) for high-frequency observations and different sample sizes (for $N=20$ in red, for $N=50$ in green, for $N=100$ in black). Here one can see that even for small sample sizes $(N=20)$, estimates are centered around the true parameter (negligible bias, though with considerable variance) and their distribution approaches a normal distribution for large $N$, becoming increasingly centered around the true value (zero bias). \looseness=-1  
\vspace{-0.15in}
\begin{figure}[h!]
	\centering
	\includegraphics[width = 0.45 \textwidth]{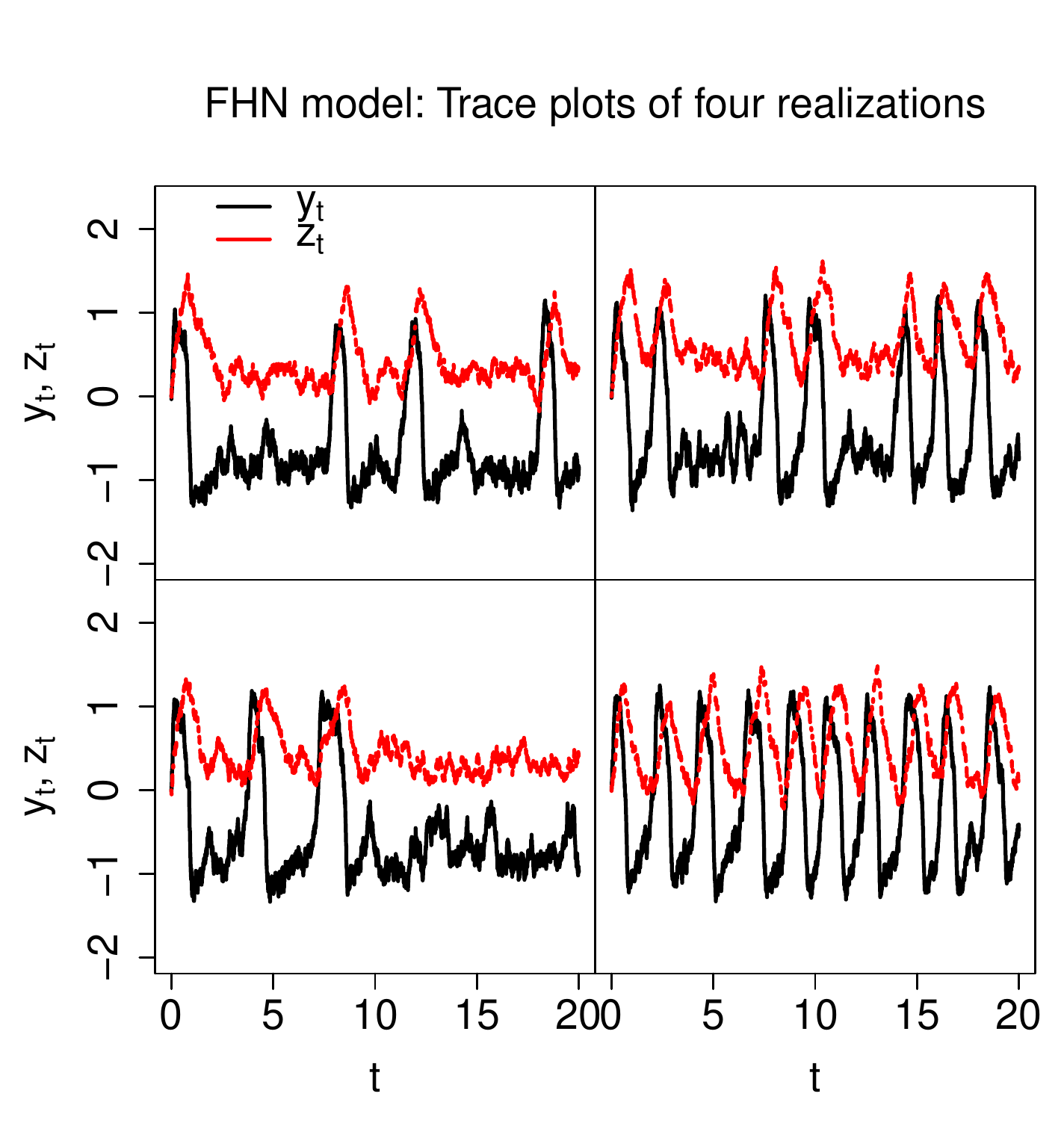}
	\caption{	\label{fig:trace_plots_FHN}\footnotesize FHN model. Trace plots of four realizations of the stochastic FHN model with random effects. The corresponding (rounded) realized parameter values of $\mu + \phi^i$ are $(8.86,4.29,1.50,1.39), (10.16,7.49,1.73,1.40), (9.98,5.49,1.10,1.07), (9.26,5.17, 1.84, 1.01 )$.}
\end{figure}

{\footnotesize
	{\renewcommand{\baselinestretch}{0.6} 
\begin{table}	
	\caption{\label{tab:Ex2_1} \footnotesize FHN model. Shown are estimated relative bias and RMSE of  $\hat{\mu}$ and diag$(\hat{\Omega})$. 
		The sample size is fixed to $N=50$, but different sampling intervals are considered ($\Delta t = 0.001, 0.01, 0.1$). 
		For each value of $\Delta t$, the estimation was repeated on $M=500$ generated data sets.}	
	\begin{tabular}{r r@{\hskip 0.5cm} rr r@{\hskip 0.5cm} rr r@{\hskip 0.5cm} rr}
		\\
		\hline\\
		& &\multicolumn{2}{c}{$\Delta t=0.001$} & & \multicolumn{2}{c}{$\Delta t=0.01$} & & \multicolumn{2}{c}{$\Delta t = 0.1$}\\[0.2cm]
		\cline{3-4}   \cline{6-7}  \cline{9-10}\\
		true value && rel. bias & RMSE && rel. bias & RMSE && rel. bias & RMSE \\[0.2cm]
		\hline\\[-0.1cm]
$\varepsilon\quad$ 0.10 &  & 0.003 & 0.022 &  & 0.030 & 0.037 &  & 0.356 & 0.356 \\ 
$s\quad$0.50 &  & 0.001 & 0.033 &  & 0.002 & 0.033 &  & 0.009 & 0.035 \\ 
$\gamma\quad$1.50 &  & -0.000 & 0.031 &  & -0.006 & 0.032 &  & -0.079 & 0.121 \\ 
$\eta\quad$1.20 &  & -0.001 & 0.031 &  & -0.005 & 0.032 &  & -0.051 & 0.067 \\ 
$1/\varepsilon\quad$10.00 &  & -0.003 & 0.216 &  & -0.028 & 0.349 &  & -0.262 & 2.624 \\ 
$s/\varepsilon\quad$5.00 &  & -0.002 & 0.135 &  & -0.026 & 0.188 &  & -0.256 & 1.283 \\[0.2cm]
2.25 &  & -0.048 & 0.469 &  & -0.155 & 0.539 &  & -0.690 & 1.563 \\ 
diag$(\Omega)\quad$1.00 &  & -0.025 & 0.197 &  & -0.044 & 0.197 &  & -0.281 & 0.318 \\ 
0.04 &  & -0.035 & 0.010 &  & -0.044 & 0.010 &  & -0.212 & 0.012 \\ 
0.04 &  & -0.007 & 0.010 &  & -0.028 & 0.009 &  & -0.218 & 0.012 \\[0.2cm]
		\hline
	\end{tabular}
\end{table}}}

\begin{figure}[h!]
	\centering
	\includegraphics[width = 0.4\textwidth]{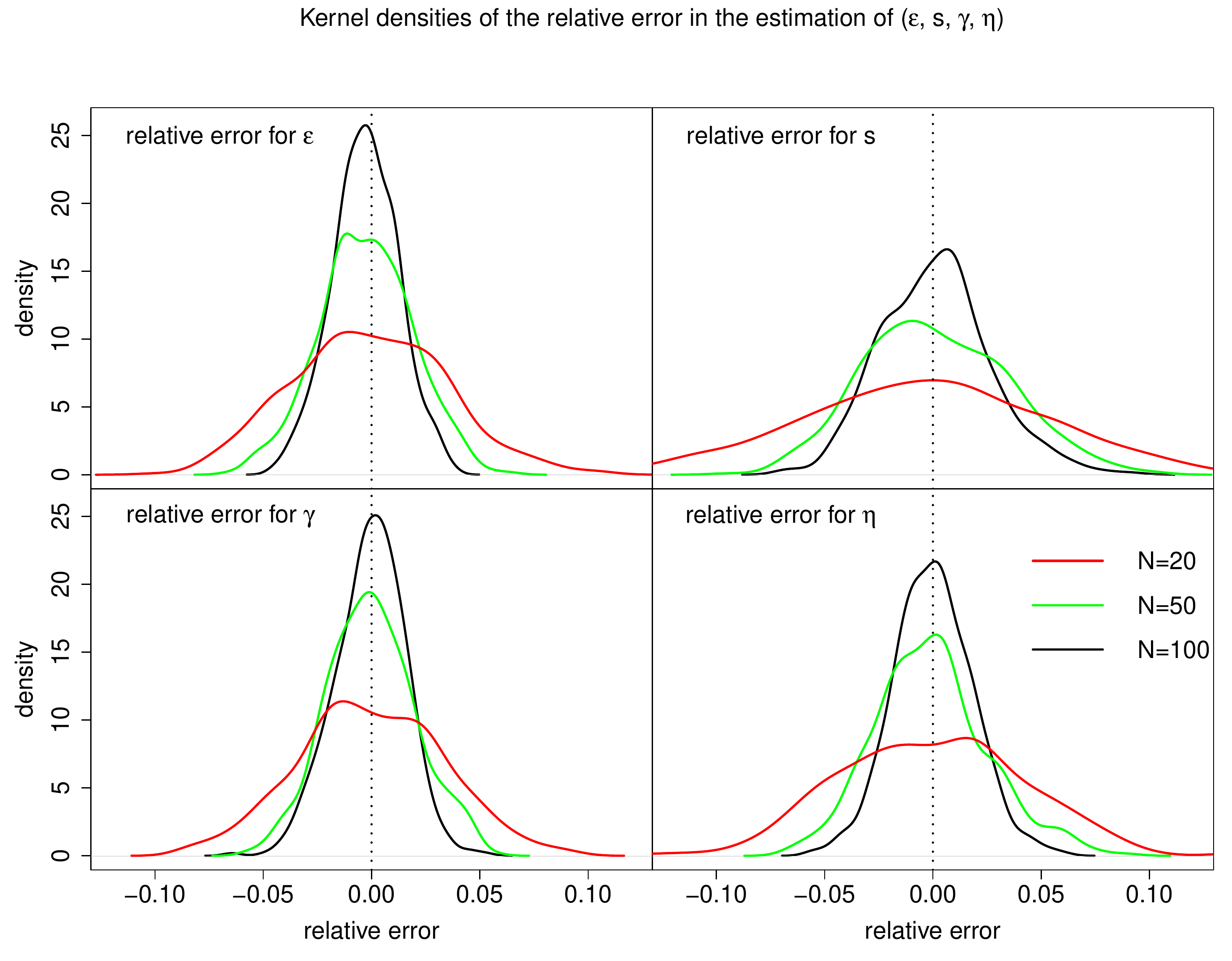}
	\caption{	\label{fig:FHN_relErr} \footnotesize FHN model. 
		Density plots of the relative error err$_{\nu_j}(N) = \frac{\hat{\nu}_j^{(m)} - \nu_j}{\nu_j}, m=1,\ldots,M,$ of the estimation of fixed effects, given here on the original $\nu$-scale. The sampling interval is $\Delta t = 0.001$, but different sample sizes are considered: $N=20$ (red), $N=50$ (green), $N=100$ (black). Estimation was repeated on $M=500$ generated data sets. }
\end{figure}

\section{Discussion}
Stochastic differential equations with random effects constitute an attractive class of statistical models, especially  for biological data. We extended the framework introduced in \cite{delattre2013maximum} to multidimensional and time-inhomogeneous state processes and proved the consistency of the maximum likelihood estimator using standard tools when the random effects enter the model linearly. However, the dynamics are allowed to be non-linear in the state. This particular setting comprises numerous well-known models, such as the predator-prey (or Lotka-Volterra) model \citep{murray2002mathematical}, the Lorenz equations introduced by \citet{lorenz1963deterministic}, which have been used to model, e.g.,  temperature, wind speed and humidity, the Brusselator model \cite[19.4]{kondepudi2014modern}, the  FHN model (see section \ref{simulations}) or the SIR (susceptible-infected-removed) model introduced by \cite{kermack1927contribution}, an epidemic model which has widely been studied and applied \citep{keeling2008modeling, jiang2011asymptotic, guy2015approximation}. 
 
We  examined the extension of the i.i.d. model to independent, but not identically distributed observations, with particular emphasis on the inclusion of covariate information. We pointed out the fundamental difference to  assumptions on the covariates that are standard in regression analysis,  gave  conditions for asymptotic normality of the MLE  (or, more generally, the LAN of the models) when the observations do not come from the same distribution and illustrated their verification by means of an example. \looseness=-1

The quality of the estimation in terms of sample size and sampling frequency was investigated in two simulation studies. In the first one, we use a model with covariates, which is linear in parameters and state. When observations are sampled at high frequency, estimation results were  convincing already for small sample sizes ($N=20$), despite the comparably large number (11 fixed effects and 6 variances) of unknown parameters. A moderate sampling interval (of $\Delta t = 0.01$) still gave  good results for all considered sample sizes. However, when sampling at low frequency ($\Delta t = 0.1$), the discrete-time bias  makes itself felt. Similar observations could be made in the second simulation setting. Here, we considered a stochastic FHN model, which is linear in the parameter, but non-linear in the state. It is an example of a diffusion which violates common assumptions on the growth of the drift function of a diffusion, which ensure the validity of many results (existence of solution, absolute continuity of measures, convergence of the time-discretized estimator to its continuous-time analogue). The estimation results are  accurate for high-frequency data, even at small sample sizes ($N=20$), and also for moderate frequency data, the estimation was still  convincing (except for the variance of the time-scale random effect). A considerable drop in accuracy occurred for low-frequency data (sampling interval of $\Delta t = 0.1$). If this method is to be used on such kind of data, algorithms for data imputation should be  applied prior to estimation, in order to reduce the discrete-time bias. 

The asymptotic normality of the MLE lends itself naturally to hypothesis testing of parameters by means of the Wald test. For the linear transfer model we estimate the false-positive rate, revealing a slight liberalism of the test procedure, and compute the test's power for different "true" values of parameters. \looseness=-1

We have only studied the method's applicability to models with up to 17 parameters. Even in the case of an explicit likelihood, the MLE of the (unknown) covariance matrix of the random effects vector is implicit and estimation requires numerical optimization, which may hamper estimation when the parameter space has a high dimension. Another drawback is the already mentioned inherent discrete-time bias of the estimation procedure. It is negligible if observations are sampled at sufficiently high frequency, but for low-frequency observations, a severe bias occurs (cf. simulation results), which is to bear in mind in applications. A possible solution could be to impute data at time points in between observation times, and conduct the estimation on the enlarged data set \citep{bladt2016simulation}. Related to that is the problem of incomplete observations, where only some of the coordinates in the state space are observed, and an entire path of a completely unobserved (latent) coordinate should be inferred \citep{berg2013synaptic, ditlevsen2014estimation}. Missing observations of one or more coordinates is not untypical for biological data. This, at a first step, prohibits application of the proposed estimation procedure, as it relies on the assumption of complete data observations. Such statistical recovery of hidden state coordinates remains a topic for future research.

\section*{Acknowledgments}
The work is part of the Dynamical Systems Interdisciplinary Network, University
of Copenhagen. Adeline Samson has been partially supported by the LabExPERSYVAL-Lab (ANR-11-LABX-0025-01).

\footnotesize
\bibliographystyle{chicago}
\bibliography{/Users/mareile/Dropbox/Forskning/InferenceSDEs/parameter_estimation2}

\begin{thebibliography}{}

\bibitem[\protect\citeauthoryear{A{\"\i}t-Sahalia}{A{\"\i}t-Sahalia}{2002}]{ait2002maximum}
A{\"\i}t-Sahalia, Y. (2002).
\newblock Maximum likelihood estimation of discretely sampled diffusions: A
  closed-form approximation approach.
\newblock {\em Econometrica\/}~{\em 70\/}(1), 223--262.

\bibitem[\protect\citeauthoryear{Beal and Sheiner}{Beal and
  Sheiner}{1981}]{beal1981estimating}
Beal, S.~L. and L.~B. Sheiner (1981).
\newblock Estimating population kinetics.
\newblock {\em Critical Reviews in Biomedical Engineering\/}~{\em 8\/}(3),
  195--222.

\bibitem[\protect\citeauthoryear{Berg and Ditlevsen}{Berg and
  Ditlevsen}{2013}]{berg2013synaptic}
Berg, R.~W. and S.~Ditlevsen (2013).
\newblock Synaptic inhibition and excitation estimated via the time constant of
  membrane potential fluctuations.
\newblock {\em Journal of Neurophysiology\/}~{\em 110\/}(4), 1021--1034.

\bibitem[\protect\citeauthoryear{Bladt, Finch, and S{\o}rensen}{Bladt
  et~al.}{2016}]{bladt2016simulation}
Bladt, M., S.~Finch, and M.~S{\o}rensen (2016).
\newblock Simulation of multivariate diffusion bridges.
\newblock {\em Journal of the Royal Statistical Society: Series B (Statistical
  Methodology)\/}~{\em 78\/}(2), 343--369.

\bibitem[\protect\citeauthoryear{Bradley and Gart}{Bradley and
  Gart}{1962}]{bradley1962asymptotic}
Bradley, R.~A. and J.~J. Gart (1962).
\newblock The asymptotic properties of {ML} estimators when sampling from
  associated populations.
\newblock {\em Biometrika\/}~{\em 49\/}(1/2), 205--214.

\bibitem[\protect\citeauthoryear{Davidian and Giltinan}{Davidian and
  Giltinan}{2003}]{davidian2003nonlinear}
Davidian, M. and D.~M. Giltinan (2003).
\newblock Nonlinear models for repeated measurement data: an overview and
  update.
\newblock {\em Journal of Agricultural, Biological, and Environmental
  Statistics\/}~{\em 8\/}(4), 387--419.

\bibitem[\protect\citeauthoryear{Delattre, Genon-Catalot, and Samson}{Delattre
  et~al.}{2013}]{delattre2013maximum}
Delattre, M., V.~Genon-Catalot, and A.~Samson (2013).
\newblock Maximum likelihood estimation for stochastic differential equations
  with random effects.
\newblock {\em Scandinavian Journal of Statistics\/}~{\em 40\/}(2), 322--343.

\bibitem[\protect\citeauthoryear{Delattre, Genon-Catalot, and Samson}{Delattre
  et~al.}{2015}]{delattre2015estimation}
Delattre, M., V.~Genon-Catalot, and A.~Samson (2015).
\newblock Estimation of population parameters in stochastic differential
  equations with random effects in the diffusion coefficient.
\newblock {\em ESAIM: Probability and Statistics\/}~{\em 19}, 671--688.

\bibitem[\protect\citeauthoryear{Delattre and Lavielle}{Delattre and
  Lavielle}{2013}]{delattre2013coupling}
Delattre, M. and M.~Lavielle (2013).
\newblock Coupling the {SAEM} algorithm and the extended {Kalman} filter for
  maximum likelihood estimation in mixed-effects diffusion models.
\newblock {\em Statistics and its Interface\/}~{\em 6\/}(4), 519--532.

\bibitem[\protect\citeauthoryear{Delyon, Lavielle, and Moulines}{Delyon
  et~al.}{1999}]{delyon1999convergence}
Delyon, B., M.~Lavielle, and E.~Moulines (1999).
\newblock Convergence of a stochastic approximation version of the {EM}
  algorithm.
\newblock {\em Annals of Statistics\/}, 94--128.

\bibitem[\protect\citeauthoryear{Ditlevsen and De~Gaetano}{Ditlevsen and
  De~Gaetano}{2005}]{ditlevsen2005mixed}
Ditlevsen, S. and A.~De~Gaetano (2005).
\newblock Mixed effects in stochastic differential equation models.
\newblock {\em REVSTAT-Statistical Journal\/}~{\em 3\/}(2), 137--153.

\bibitem[\protect\citeauthoryear{Ditlevsen, Samson, et~al.}{Ditlevsen
  et~al.}{2014}]{ditlevsen2014estimation}
Ditlevsen, S., A.~Samson, et~al. (2014).
\newblock Estimation in the partially observed stochastic morris--lecar
  neuronal model with particle filter and stochastic approximation methods.
\newblock {\em The Annals of Applied Statistics\/}~{\em 8\/}(2), 674--702.

\bibitem[\protect\citeauthoryear{Ditlevsen, Yip, and
  Holstein-Rathlou}{Ditlevsen et~al.}{2005}]{ditlevsen2005parameter}
Ditlevsen, S., K.-P. Yip, and N.-H. Holstein-Rathlou (2005).
\newblock Parameter estimation in a stochastic model of the tubuloglomerular
  feedback mechanism in a rat nephron.
\newblock {\em Mathematical Biosciences\/}~{\em 194\/}(1), 49--69.

\bibitem[\protect\citeauthoryear{Donnet, Foulley, and Samson}{Donnet
  et~al.}{2010}]{donnet2010bayesian}
Donnet, S., J.-L. Foulley, and A.~Samson (2010).
\newblock Bayesian analysis of growth curves using mixed models defined by
  stochastic differential equations.
\newblock {\em Biometrics\/}~{\em 66\/}(3), 733--741.

\bibitem[\protect\citeauthoryear{Donnet, Samson, et~al.}{Donnet
  et~al.}{2008}]{donnet2008parametric}
Donnet, S., A.~Samson, et~al. (2008).
\newblock Parametric inference for mixed models defined by stochastic
  differential equations.
\newblock {\em ESAIM P\&S\/}~{\em 12}, 196--218.

\bibitem[\protect\citeauthoryear{Durham and Gallant}{Durham and
  Gallant}{2002}]{durham2002numerical}
Durham, G.~B. and A.~R. Gallant (2002).
\newblock Numerical techniques for maximum likelihood estimation of
  continuous-time diffusion processes.
\newblock {\em Journal of Business \& Economic Statistics\/}~{\em 20\/}(3),
  297--338.

\bibitem[\protect\citeauthoryear{FitzHugh}{FitzHugh}{1955}]{fitzhugh1955mathematical}
FitzHugh, R. (1955).
\newblock Mathematical models of threshold phenomena in the nerve membrane.
\newblock {\em The Bulletin of Mathematical Biophysics\/}~{\em 17\/}(4),
  257--278.

\bibitem[\protect\citeauthoryear{Gabbay, Thagard, Woods, Bandyopadhyay, and
  Forster}{Gabbay et~al.}{2011}]{gabbay2011philosophy}
Gabbay, D.~M., P.~Thagard, J.~Woods, P.~S. Bandyopadhyay, and M.~R. Forster
  (2011).
\newblock {\em Philosophy of statistics}, Volume~7.
\newblock Elsevier.

\bibitem[\protect\citeauthoryear{Gro{\ss}e~Ruse, S{\o}ndergaard, Ditlevsen,
  Damgaard, Fuglsang, Ottesen, and Madsen}{Gro{\ss}e~Ruse
  et~al.}{2015}]{ruse2015absorption}
Gro{\ss}e~Ruse, M., L.~R. S{\o}ndergaard, S.~Ditlevsen, M.~Damgaard,
  S.~Fuglsang, J.~T. Ottesen, and J.~L. Madsen (2015).
\newblock Absorption and initial metabolism of 75 se-l-selenomethionine: a
  kinetic model based on dynamic scintigraphic data.
\newblock {\em British Journal of Nutrition\/}~{\em 114\/}(10), 1718--1723.

\bibitem[\protect\citeauthoryear{Guedj, Thi{\'e}baut, and Commenges}{Guedj
  et~al.}{2007}]{guedj2007maximum}
Guedj, J., R.~Thi{\'e}baut, and D.~Commenges (2007).
\newblock Maximum likelihood estimation in dynamical models of {HIV}.
\newblock {\em Biometrics\/}~{\em 63\/}(4), 1198--1206.

\bibitem[\protect\citeauthoryear{Guy, Lar{\'e}do, and Vergu}{Guy
  et~al.}{2015}]{guy2015approximation}
Guy, R., C.~Lar{\'e}do, and E.~Vergu (2015).
\newblock Approximation of epidemic models by diffusion processes and their
  statistical inference.
\newblock {\em Journal of Mathematical Biology\/}~{\em 70\/}(3), 621--646.

\bibitem[\protect\citeauthoryear{Hoadley}{Hoadley}{1971}]{hoadley1971asymptotic}
Hoadley, B. (1971).
\newblock Asymptotic properties of maximum likelihood estimators for the
  independent not identically distributed case.
\newblock {\em The Annals of Mathematical Statistics\/}, 1977--1991.

\bibitem[\protect\citeauthoryear{Hodgkin and Huxley}{Hodgkin and
  Huxley}{1952}]{hodgkin1952quantitative}
Hodgkin, A.~L. and A.~F. Huxley (1952).
\newblock A quantitative description of membrane current and its application to
  conduction and excitation in nerve.
\newblock {\em The Journal of Physiology\/}~{\em 117\/}(4), 500.

\bibitem[\protect\citeauthoryear{Ibragimov and Has'minskii}{Ibragimov and
  Has'minskii}{2013}]{ibragimov2013statistical}
Ibragimov, I.~A. and R.~Z. Has'minskii (2013).
\newblock {\em Statistical Estimation: Asymptotic Theory}, Volume~16.
\newblock Springer Science.

\bibitem[\protect\citeauthoryear{Jensen, Ditlevsen, Kessler, and
  Papaspiliopoulos}{Jensen et~al.}{2012}]{jensen2012markov}
Jensen, A.~C., S.~Ditlevsen, M.~Kessler, and O.~Papaspiliopoulos (2012).
\newblock Markov chain monte carlo approach to parameter estimation in the
  fitzhugh-nagumo model.
\newblock {\em Physical Review E\/}~{\em 86\/}(4), 041114.

\bibitem[\protect\citeauthoryear{Jiang, Yu, Ji, and Shi}{Jiang
  et~al.}{2011}]{jiang2011asymptotic}
Jiang, D., J.~Yu, C.~Ji, and N.~Shi (2011).
\newblock Asymptotic behavior of global positive solution to a stochastic {SIR}
  model.
\newblock {\em Mathematical and Computer Modelling\/}~{\em 54\/}(1), 221--232.

\bibitem[\protect\citeauthoryear{Keeling and Rohani}{Keeling and
  Rohani}{2008}]{keeling2008modeling}
Keeling, M.~J. and P.~Rohani (2008).
\newblock {\em Modeling infectious diseases in humans and animals}.
\newblock Princeton University Press.

\bibitem[\protect\citeauthoryear{Kermack and McKendrick}{Kermack and
  McKendrick}{1927}]{kermack1927contribution}
Kermack, W.~O. and A.~G. McKendrick (1927).
\newblock A contribution to the mathematical theory of epidemics.
\newblock In {\em Proceedings of the Royal Society of London A: Mathematical,
  Physical and Engineering Sciences}, Volume 115, pp.\  700--721. The Royal
  Society.

\bibitem[\protect\citeauthoryear{Kessler, Lindner, and S{\o}rensen}{Kessler
  et~al.}{2012}]{kessler2012statistical}
Kessler, M., A.~Lindner, and M.~S{\o}rensen (2012).
\newblock {\em Statistical methods for stochastic differential equations}.
\newblock CRC Press.

\bibitem[\protect\citeauthoryear{Klim, Mortensen, Kristensen, Overgaard, and
  Madsen}{Klim et~al.}{2009}]{klim2009population}
Klim, S., S.~B. Mortensen, N.~R. Kristensen, R.~V. Overgaard, and H.~Madsen
  (2009).
\newblock Population stochastic modelling ({PSM})---an {R} package for
  mixed-effects models based on stochastic differential equations.
\newblock {\em Computer Methods and Programs in Biomedicine\/}~{\em 94\/}(3),
  279--289.

\bibitem[\protect\citeauthoryear{Kondepudi and Prigogine}{Kondepudi and
  Prigogine}{2014}]{kondepudi2014modern}
Kondepudi, D. and I.~Prigogine (2014).
\newblock {\em Modern thermodynamics: from heat engines to dissipative
  structures}.
\newblock John Wiley \&amp; Sons.

\bibitem[\protect\citeauthoryear{Lavielle}{Lavielle}{2014}]{lavielle2014mixed}
Lavielle, M. (2014).
\newblock {\em Mixed effects models for the population approach: models, tasks,
  methods and tools}.
\newblock CRC Press.

\bibitem[\protect\citeauthoryear{Le~Cam}{Le~Cam}{2012}]{le2012asymptotic}
Le~Cam, L. (2012).
\newblock {\em Asymptotic methods in statistical decision theory}.
\newblock Springer Science, New York.

\bibitem[\protect\citeauthoryear{Leander, Almquist, Ahlstr{\"o}m, Gabrielsson,
  and Jirstrand}{Leander et~al.}{2015}]{leander2015mixed}
Leander, J., J.~Almquist, C.~Ahlstr{\"o}m, J.~Gabrielsson, and M.~Jirstrand
  (2015).
\newblock Mixed effects modeling using stochastic differential equations:
  illustrated by pharmacokinetic data of nicotinic acid in obese zucker rats.
\newblock {\em The AAPS Journal\/}~{\em 17\/}(3), 586--596.

\bibitem[\protect\citeauthoryear{Leander, Lundh, and Jirstrand}{Leander
  et~al.}{2014}]{leander2014stochastic}
Leander, J., T.~Lundh, and M.~Jirstrand (2014).
\newblock Stochastic differential equations as a tool to regularize the
  parameter estimation problem for continuous time dynamical systems given
  discrete time measurements.
\newblock {\em Mathematical Biosciences\/}~{\em 251}, 54--62.

\bibitem[\protect\citeauthoryear{Lehmann and Romano}{Lehmann and
  Romano}{2006}]{lehmann2006testing}
Lehmann, E.~L. and J.~P. Romano (2006).
\newblock {\em Testing statistical hypotheses}.
\newblock Springer Science \& Business Media.

\bibitem[\protect\citeauthoryear{Lindstrom and Bates}{Lindstrom and
  Bates}{1990}]{lindstrom1990nonlinear}
Lindstrom, M.~J. and D.~M. Bates (1990).
\newblock Nonlinear mixed effects models for repeated measures data.
\newblock {\em Biometrics\/}, 673--687.

\bibitem[\protect\citeauthoryear{Lo}{Lo}{1988}]{lo1988maximum}
Lo, A.~W. (1988).
\newblock Maximum likelihood estimation of generalized {I}t{\^o} processes with
  discretely sampled data.
\newblock {\em Econometric Theory\/}~{\em 4\/}(2), 231--247.

\bibitem[\protect\citeauthoryear{Lorenz}{Lorenz}{1963}]{lorenz1963deterministic}
Lorenz, E.~N. (1963).
\newblock Deterministic nonperiodic flow.
\newblock {\em Journal of the atmospheric sciences\/}~{\em 20\/}(2), 130--141.

\bibitem[\protect\citeauthoryear{M{\o}ller, Overgaard, Madsen, Hansen,
  Pedersen, and Ingwersen}{M{\o}ller et~al.}{2010}]{moller2010}
M{\o}ller, J.~B., R.~V. Overgaard, H.~Madsen, T.~Hansen, O.~Pedersen, and S.~H.
  Ingwersen (2010).
\newblock Predictive performance for population models using stochastic
  differential equations applied on data from an oral glucose tolerance test.
\newblock {\em Journal of Pharmacokinetics and Pharmacodynamics\/}~{\em
  37\/}(1), 85--98.

\bibitem[\protect\citeauthoryear{Mortensen, Klim, Dammann, Kristensen, Madsen,
  and Overgaard}{Mortensen et~al.}{2007}]{mortensen2007matlab}
Mortensen, S.~B., S.~Klim, B.~Dammann, N.~R. Kristensen, H.~Madsen, and R.~V.
  Overgaard (2007).
\newblock A matlab framework for estimation of {NLME} models using stochastic
  differential equations.
\newblock {\em Journal of Pharmacokinetics and Pharmacodynamics\/}~{\em
  34\/}(5), 623--642.

\bibitem[\protect\citeauthoryear{Murray}{Murray}{2002}]{murray2002mathematical}
Murray, J.~D. (2002).
\newblock {\em Mathematical Biology I: An Introduction}, Volume~17 of {\em
  Interdisciplinary Applied Mathematics}.
\newblock Springer, New York, NY, USA,.

\bibitem[\protect\citeauthoryear{Nagumo, Arimoto, and Yoshizawa}{Nagumo
  et~al.}{1962}]{nagumo1962active}
Nagumo, J., S.~Arimoto, and S.~Yoshizawa (1962).
\newblock An active pulse transmission line simulating nerve axon.
\newblock {\em Proceedings of the IRE\/}~{\em 50\/}(10), 2061--2070.

\bibitem[\protect\citeauthoryear{Overgaard, Jonsson, Torn{\o}e, and
  Madsen}{Overgaard et~al.}{2005}]{overgaard2005non}
Overgaard, R.~V., N.~Jonsson, C.~W. Torn{\o}e, and H.~Madsen (2005).
\newblock Non-linear mixed-effects models with stochastic differential
  equations: implementation of an estimation algorithm.
\newblock {\em Journal of pharmacokinetics and pharmacodynamics\/}~{\em
  32\/}(1), 85--107.

\bibitem[\protect\citeauthoryear{Pedersen}{Pedersen}{1995}]{pedersen1995new}
Pedersen, A.~R. (1995).
\newblock A new approach to maximum likelihood estimation for stochastic
  differential equations based on discrete observations.
\newblock {\em Scandinavian journal of statistics\/}, 55--71.

\bibitem[\protect\citeauthoryear{Phillips and Yu}{Phillips and
  Yu}{2009}]{phillips2009maximum}
Phillips, P.~C. and J.~Yu (2009).
\newblock Maximum likelihood and gaussian estimation of continuous time models
  in finance.
\newblock In {\em Handbook of financial time series}, pp.\  497--530. Springer,
  New York.

\bibitem[\protect\citeauthoryear{Picchini, De~Gaetano, and Ditlevsen}{Picchini
  et~al.}{2010}]{picchini2010stochastic}
Picchini, U., A.~De~Gaetano, and S.~Ditlevsen (2010).
\newblock Stochastic differential mixed-effects models.
\newblock {\em Scandinavian Journal of Statistics\/}~{\em 37\/}(1), 67--90.

\bibitem[\protect\citeauthoryear{Picchini and Ditlevsen}{Picchini and
  Ditlevsen}{2011}]{picchini2011practical}
Picchini, U. and S.~Ditlevsen (2011).
\newblock Practical estimation of high dimensional stochastic differential
  mixed-effects models.
\newblock {\em Computational Statistics \& Data Analysis\/}~{\em 55\/}(3),
  1426--1444.

\bibitem[\protect\citeauthoryear{Pinheiro and Bates}{Pinheiro and
  Bates}{2006}]{pinheiro2006mixed}
Pinheiro, J. and D.~Bates (2006).
\newblock {\em Mixed-effects models in S and S-PLUS}.
\newblock Springer Science, New York.

\bibitem[\protect\citeauthoryear{Ribba, Holford, Magni, Troc{\'o}niz,
  Gueorguieva, Girard, Sarr, Elishmereni, Kloft, and Friberg}{Ribba
  et~al.}{2014}]{ribba2014review}
Ribba, B., N.~H. Holford, P.~Magni, I.~Troc{\'o}niz, I.~Gueorguieva, P.~Girard,
  C.~Sarr, M.~Elishmereni, C.~Kloft, and L.~E. Friberg (2014).
\newblock A review of mixed-effects models of tumor growth and effects of
  anticancer drug treatment used in population analysis.
\newblock {\em CPT: Pharmacometrics \& Systems Pharmacology\/}~{\em 3\/}(5),
  1--10.

\bibitem[\protect\citeauthoryear{Torn{\o}e, Agers{\o}, Jonsson, Madsen, and
  Nielsen}{Torn{\o}e et~al.}{2004}]{tornoe2004non}
Torn{\o}e, C.~W., H.~Agers{\o}, E.~N. Jonsson, H.~Madsen, and H.~A. Nielsen
  (2004).
\newblock Non-linear mixed-effects pharmacokinetic/pharmacodynamic modelling in
  {NLME} using differential equations.
\newblock {\em Computer Methods and Programs in Biomedicine\/}~{\em 76\/}(1),
  31--40.

\bibitem[\protect\citeauthoryear{Torn{\o}e, Overgaard, Agers{\o}, Nielsen,
  Madsen, and Jonsson}{Torn{\o}e et~al.}{2005}]{tornoe2005stochastic}
Torn{\o}e, C.~W., R.~V. Overgaard, H.~Agers{\o}, H.~A. Nielsen, H.~Madsen, and
  E.~N. Jonsson (2005).
\newblock Stochastic differential equations in {NONMEM}: implementation,
  application, and comparison with ordinary differential equations.
\newblock {\em Pharmaceutical Research\/}~{\em 22\/}(8), 1247--1258.

\bibitem[\protect\citeauthoryear{Van~der Vaart}{Van~der
  Vaart}{2000}]{van2000asymptotic}
Van~der Vaart, A.~W. (2000).
\newblock {\em Asymptotic statistics}, Volume~3.
\newblock Cambridge University Press.

\bibitem[\protect\citeauthoryear{Wang}{Wang}{2007}]{wang2007algorithms}
Wang, J. (2007).
\newblock {EM} algorithms for nonlinear mixed effects models.
\newblock {\em Computational Statistics \& Data Analysis\/}~{\em 51\/}(6),
  3244--3256.

\bibitem[\protect\citeauthoryear{Wolfinger}{Wolfinger}{1993}]{wolfinger1993laplace}
Wolfinger, R. (1993).
\newblock Laplace's approximation for nonlinear mixed models.
\newblock {\em Biometrika\/}~{\em 80\/}(4), 791--795.

\end{thebibliography}
\end{document}